%% file: arxiv.tex
\documentclass{article}
\usepackage{authblk}
\usepackage{vmargin}
\setmarginsrb{20mm}{20mm}{20mm}{20mm}{0pt}{0mm}{0pt}{0mm}

\usepackage{etex}      
\usepackage{microtype}		
\usepackage{amsmath}		
\usepackage{amssymb}		

\usepackage{cmll}
\usepackage{stmaryrd}
\usepackage{proof}
\usepackage[all]{xy}

\usepackage{wrapfig}
\usepackage{tikz}
\usetikzlibrary{calc, intersections, shapes.geometric}
\tikzset{every picture/.style={
scale=0.15,
font={\fontsize{8pt}{12}\selectfont},
baseline={(current bounding box.west)}
}} 

\usepackage{amsthm}		
\theoremstyle{definition}
\newtheorem{definition}{Definition}[section] 

\theoremstyle{plain}
\newtheorem{proposition}[definition]{Proposition} 
 \newtheorem{lemma}[definition]{Lemma}		  
 \newtheorem{theorem}[definition]{Theorem}	  
 \newtheorem{corollary}[definition]{Corollary}	  
 
\newtheorem*{notation*}{Notation}
 \newcommand{\thistheoremname}{}
 \newtheorem*{generic*}{\thistheoremname}
 \newenvironment{freenamed*}[1]
   {\renewcommand{\thistheoremname}{#1}%
    \begin{generic*}}
   {\end{generic*}}
\theoremstyle{remark}

\newcommand{\abs}[3]{\lambda {#1}^{#2}. #3}
\newcommand{\app}[2]{#1 \, #2}	      

\newcommand{\FV}{\mathrm{FV}}
\newcommand{\plug}[2]{#1 \langle #2 \rangle}
\newcommand{\emptyCtxt}{\langle \cdot \rangle}

\newcommand{\rewritesfirst}{$_\rightarrowtriangle$}

\newcommand{\N}{\mathbb{N}}

\newcommand{\iffdef}{\overset{\text{def.}}{\iff}}
\newcommand{\twoheadmapsto}{
\mathrel{\ooalign{$\twoheadrightarrow$\cr%
\kern-.15ex\raisebox{.2ex}{\scalebox{1}[0.8]{$\shortmid$}}\cr}}}
\newcommand{\longtwoheadmapsto}{
\mathrel{\ooalign{$\longtwoheadrightarrow$\cr%
\kern-.15ex\raisebox{.2ex}{\scalebox{1}[0.8]{$\shortmid$}}\cr}}}

\newcommand{\longequal}{=\joinrel=}

\input graph.tex
\input picture.tex

\newif\ifshowFigGenerators
\newif\ifshowFigBox
\newif\ifshowFigPassTransitions
\newif\ifshowFigRewriteTransitions
\newif\ifshowDefTranslationsTerm
\newif\ifshowLemTranslationComposition
\newif\ifshowFigTranslationTerms
\newif\ifshowFigTranslationEvalCtxts
\newif\ifshowDefSimulation
\newif\ifshowLemTranslationEvalCtxt
\newif\ifshowPrfDecomposeTranslationEvalCtxt
\newif\ifshowIllustrateSimulation
\showFigGeneratorstrue
\showFigBoxtrue
\showFigPassTransitionstrue
\showFigRewriteTransitionstrue
\showDefTranslationsTermtrue
\showLemTranslationCompositiontrue
\showFigTranslationTermstrue
\showFigTranslationEvalCtxtstrue
\showDefSimulationtrue
\showLemTranslationEvalCtxttrue
\showPrfDecomposeTranslationEvalCtxttrue
\showIllustrateSimulationtrue

\title{The Dynamic Geometry of Interaction Machine:  \\
A Call-by-need Graph Rewriter}

\author{Koko Muroya and Dan Ghica\\ University of Birmingham, UK}
\date{}

\begin{document}

\maketitle

\begin{abstract}
 Girard's Geometry of Interaction (GoI), a semantics designed for linear logic
 proofs, has been also successfully applied to programming language semantics.
 One way is to use abstract machines that pass a token
 on a fixed graph along a path indicated by the GoI.
 These token-passing abstract machines are space
 efficient, because they handle duplicated computation by repeating
 the same moves of a token on the fixed graph.
 Although they can be adapted to obtain sound models with regard to the equational theories of various  evaluation strategies for the lambda calculus, it can be at the expense of significant time costs.
 In this paper we show a token-passing abstract machine that can implement
 evaluation strategies for the lambda calculus, with certified time efficiency.
 Our abstract machine, called the \textit{Dynamic GoI Machine} (DGoIM), rewrites the graph to avoid replicating computation, using the token to find the redexes.
 The flexibility of interleaving token transitions and graph rewriting allows the DGoIM to balance the
 trade-off of space and time costs.
 This paper shows that the DGoIM can implement call-by-need evaluation for the lambda calculus by
 using a strategy of interleaving token passing with as much graph rewriting as possible.
 Our quantitative analysis confirms that the DGoIM with this strategy of
 interleaving the two kinds of possible operations on graphs can be classified as
 ``efficient'' following Accattoli's taxonomy of abstract machines.
\end{abstract}

\section{Introduction}

\subsection{Token-passing Abstract Machines for $\lambda$-calculus}

Girard's Geometry of Interaction (GoI) \cite{Girard89GoI1} is a
semantic framework for linear logic proofs~\cite{Girard87LL}.
One way of applying it to programming language semantics is via ``token-passing'' abstract machines.
A term in the $\lambda$-calculus is evaluated by representing it
as a graph, then passing a token along a path indicated by the GoI.
Token-passing  GoI  decomposes higher-order
computation into local token actions, or low-level interactions of
simple components.
It can give strikingly innovative implementation techniques for
functional programs, such as Mackie's \textit{Geometry of Implementation} compiler
\cite{Mackie95}, Ghica's \textit{Geometry of
Synthesis} (GoS) high-level synthesis tool \cite{Ghica07GoS1}, and
Sch{\"{o}}pp's resource-aware program transformation to a low-level
language \cite{Schoepp14a}.
The interaction-based approach is also convenient for the complexity
analysis of programs, e.g.\
Dal Lago and Sch{\"{o}}pp's \textsc{IntML} type system of logarithmic-space
evaluation \cite{DalLagoS16}, and Dal Lago et al.'s linear dependent
type system of polynomial-time evaluation \cite{DalLagoG11,DalLagoP12}.

Fixed-space execution is essential for GoS, since in the case of digital circuits the memory footprint of the program must be known at compile-time, and fixed. Using a restricted version of the call-by-name language Idealised Algol~\cite{GhicaS11GoS3} not only the graph, but also the token itself can be given a fixed size. Surprisingly, this technique also allows the compilation of recursive programs~\cite{GhicaSS11GoS4}. The GoS compiler shows both the usefulness of the GoI as a guideline for unconventional compilation and the natural affinity between its space-efficient abstract machine and call-by-name evaluation. The practical considerations match the prior theoretical understanding of this connection~\cite{DanosR96}.

In contrast, re-evaluating a term by repeating its token actions
poses a challenge for call-by-value evaluation
(e.g.\ \cite{FernandezM02,Schoepp14b,HoshinoMH14,DalLagoFVY15}) because
duplicated computation must not lead to repeated evaluation.
Moreover, in call-by-value repeating token actions raises the
additional technical challenge of avoiding repeating any associated computational effects
(e.g.\ \cite{Schoepp11,MuroyaHH16,DalLagoFVY17}). A partial solution to this conundrum is to focus on the soundness of the equational theory, while deliberately ignoring the time costs~\cite{MuroyaHH16}.
However,
Fern{\'{a}}ndez and Mackie suggest that in a call-by-value scenario, the time efficiency of a
token-passing abstract machine could also be improved, by allowing a token to jump along a
path, even though a time cost analysis is not given~\cite{FernandezM02}.

For us, solving the the problem of creating a GoI-style abstract machine which computes efficiently with evaluation strategies other than call-by-name is a first step in a longer-range research programme. The compilation techniques derived from the GoI can be extremely useful in the case of unconventional computational platforms. But if GoI-style techniques are to be used in a practical setting they need to extend beyond call-by-name, not just correctly but also efficiently.

\subsection{Interleaving Token Passing with Graph Rewriting}
\label{sec:IntroductionInterleaving}

A token jumping, rather than following a path, can be seen as a simple form of short-circuiting that path, which is a simple form of graph-rewriting. This idea first occurs in Mackie's work as a compiler optimisation technique~\cite{Mackie95} and is analysed in more depth theoretically by Danos and Regnier in the so-called \textit{Interaction Abstract Machine}~\cite{DanosR96}. More general graph-rewriting-based semantics have been used in a system called \textit{virtual reduction}~\cite{DanosR93}, where rewriting occurs along paths indicated by GoI, but without any token-actions. The most operational presentation of the combination of token-passing and jumping was given by Fern{\'{a}}ndez and Mackie~\cite{FernandezM02}. The interleaving of token actions and rewriting is also found in Sinot's interaction nets \cite{Sinot05,Sinot06}. We can reasonably think of the DGoIM as their abstract-machine realisation.

We build on these prior insights by adding more general, yet still efficient, graph-rewriting facilities to the setting of a GoI token-passing abstract machine.  We call an abstract machine that interleaves token passing with graph rewriting the \textit{Dynamic GoI Machine} (DGoIM), and we define it as a state transition system with transitions for token passing as well as transitions for graph
rewriting. What connects these two kinds of transitions is the token trajectory through the graph, its path. By examining it,
the DGoIM can detect redexes and trigger rewriting actions.

Through graph rewriting, the DGoIM reduces sub-graphs visited by the token,  avoiding repeated token actions and improving time efficiency.
On the other hand, graph rewriting can expand a graph by e.g.\ copying sub-graphs, so space costs can grow.
To control this trade-off of space and time cost, the DGoIM has the flexibility of interleaving token passing with graph rewriting. Once the DGoIM detects that it has traversed a redex, it may rewrite it, but it may also just propagate the token without rewriting the redex.

As a first step in our exploration of the flexibility of this machine, we consider the
two extremal cases of interleaving.
The first extremal case is ``passes-only,'' in which the DGoIM never
triggers graph rewriting, yielding an ordinary token-passing abstract
machine.
As a typical example,
the $\lambda$-term $\app{(\abs{x}{}{t})}{u}$ is evaluated like this:\\[1.5ex]
 \begin{minipage}[c]{.2\hsize}
  \centering
  \begin{tikzpicture}[baseline=(current bounding box.west)]
   \draw [rounded corners, very thick] (-6.5,0) rectangle (-0.5,3)
   node [midway] {$\abs{x}{}{t}$};
   \draw [rounded corners, very thick] (0.5,0) rectangle (6.5,3)
   node [midway] {$u$};
   \fill (0,-2) circle [radius = 0.4];
   \draw [rounded corners, very thick]
   (0,-2)--(-3.5,-1.5)--(-3.5,0)
   (0,-2)--(3.5,-1.5)--(3.5,0)
   (0,-2)--(0,-4);
  \end{tikzpicture}
 \end{minipage} \hspace{.05\hsize} %
 \begin{minipage}[c]{.65\hsize}
  \begin{enumerate}
   \item A token enters the graph on the left at the bottom open edge.
   \item A token visits and goes through the left sub-graph $\lambda x.t$.
   \item Whenever a token detects an occurrence of the variable
	 $x$ in $t$, it traverses the right sub-graph $u$, then returns carrying the resulting value.
   \item A token finally exits the graph at the bottom open
	 edge.
  \end{enumerate}
 \end{minipage} \\[1.5ex]
Step 3 is repeated whenever term $u$ needs to be re-evaluated. This strategy of interleaving corresponds to call-by-name reduction.

The other extreme is ``rewrites-first,'' in which
the DGoIM interleaves token passing with as much, and as early, graph
rewriting as possible, guided by the token. This corresponds to both call-by-value and call-by-need reductions, the difference between the two being the trajectory of the token. 
In the case of call-by-value, the token will enter the graph from the bottom, traverse the left-hand-side sub-graph, which happens to be already a value, then visit sub-graph $u$ even before $x$ is used in a call. While traversing $u$, it will cause rewrites such that when the token exits, it leaves behind the graph of a machine corresponding to a value $v$ such that $u$ reduces to $v$. The difference with call-by-need is that the token will visit $u$ only when $x$ is encountered in $\lambda x.t$.
In both cases, if repeated evaluation is required then the sub-graph corresponding now to $v$ is copied, so that one copy can be further rewritten, if needed, while the original is kept for later reference.




\subsection{Contributions}

This work presents a DGoIM model for call-by-need, which can be seen as a case study of the flexibility achieved through controlled interleaving of rewriting and token-passing. This is achieved through a rewriting strategy which turns out to be as natural as the passes-only strategy is for implementing call-by-name. The DGoIM avoids re-evaluation of a sub-term by rewriting any sub-graph visited by a token so that the updated sub-graph represents the evaluation result, but, unlike call-by-value, it starts by evaluating the sub-graph corresponding to the function $\lambda x.t$ first. We chose call-by-need mainly because of the technical challenges it poses. Adapting the technique to \textit{call-by-value} is a straightforward exercise, and we discuss other alternative in the Conclusion. 

We analyse the time cost of the DGoIM with the rewrites-first
interleaving, using Accattoli et al.'s general methodology for quantitative analysis~\cite{AccattoliBM14,Accattoli16}. Their method cannot be used ``off the shelf,''
because the DGoIM does not satisfy one of the assumptions used in \cite[Sec.~3]{Accattoli16}.
Our machine uses a more refined transition system, in which several steps correspond to a single one in \textit{loc.\ cit.}.
We overcome this technical difficulty by building a weak simulation of Danvy and
Zerny's storeless abstract machine \cite{DanvyZ13} to which the recipe does
apply. The result of the quantitative analysis confirms that the DGoIM with the rewrites-first interleaving can be classified as ``efficient,'' following Accattoli's taxonomy of abstract machines
introduced in \cite{Accattoli16}.

As we intend to use the DGoIM as a starting point for semantics-directed compilation, this result is an important confirmation that no hidden inefficiencies lurk within the fabric of the rather complex machinery of the DGoIM.

\section{The Dynamic GoI Machine}

\subsection{Well-boxed Graphs}

The graphs used to construct the DGoIM are essentially MELL proof structures
\cite{Girard87LL} of the multiplicative and exponential fragment of
linear logic.
They are directed, and built over the fixed set of nodes called
``generators'' shown in Fig.~\ref{fig:Generators}.
\begin{figure}[ht]
 \centering
 \begin{minipage}[b]{.7\hsize}
  \centering
  \figGeneratorsAx
  \figGeneratorsCut
  \figGeneratorsTensor
  \figGeneratorsParr
  \figGeneratorsPrincipal
  \figGeneratorsAuxiliary
  \figGeneratorsDereliction
  \figGeneratorsContraction
  \caption{Generators of Graphs}
  \label{fig:Generators}
 \end{minipage}%
 \begin{minipage}[b]{.25\hsize}
  \centering
  \figBox
  \caption{$\oc$-box $H$}
  \label{fig:Box}
 \end{minipage}
\end{figure}

A $\mathsf{C}_n$-node is annotated by a natural number $n$ that
indicates its in-degree, i.e.\ the number of incoming edges.
It generalises a contraction node, whose in-degree is $2$, and a
weakening node, whose in-degree is $0$, of MELL proof structures.
In Fig.~\ref{fig:Generators}, a bunch of $n$ edges is depicted by a
single arrow with a strike-out.

Graphs must satisfy the well-formedness condition below. Note that, 
unlike the usual approach~\cite{Girard87LL}, we need not
assign MELL formulas to edges, nor require a graph to
be a valid proof net.
\begin{definition}[well-boxed]
 \label{def:WellBoxed}
 A directed graph $G$ built over the generators in
 Fig.~\ref{fig:Generators} is \emph{well-boxed} if:
 \begin{itemize}
  \item it has no incoming edges
  \item each $\oc$-node $v$ in $G$ comes with a sub-graph $H$ of $G$
	and an arbitrary number of $\wn$-nodes $\vec{u}$ such that:
	\begin{itemize}
	 \item the sub-graph $H$ (called ``$\oc$-box'') is well-boxed
	       inductively and has at least one outgoing edges
	 \item the $\oc$-node $v$ (called ``principal door of $H$'')
	       is the target of one outgoing edge of $H$
	 \item the $\wn$-nodes $\vec{u}$ (called ``auxiliary doors of
	       $H$'') are the targets of all the other outgoing edges
	       of $H$
	\end{itemize}
  \item each $\wn$-node is an auxiliary door of exactly one $\oc$-box
  \item any two distinct $\oc$-boxes with distinct principal doors are
	either disjoint or nested
 \end{itemize}
\end{definition}
\noindent
Note that a $\oc$-box might have no auxiliary doors.
We use a dashed box to indicate a $\oc$-box together with its
principal door and its auxiliary doors, as in Fig.~\ref{fig:Box}.
The auxiliary doors are depicted by a single $\wn$-node with a thick
frame and with single incoming and outgoing arrows with strike-outs.
Directions of edges are omitted in the rest of the paper, if not ambiguous, to
reduce visual clutter.

\subsection{Pass Transitions and Rewrite Transitions}

The DGoIM is formalised as a labelled transition system with two kinds
of transitions, namely \emph{pass} transitions $\dashrightarrow$ and
\emph{rewrite} transitions $\rightsquigarrow$.
Labels of transitions are $\mathsf{b},\mathsf{s},\mathsf{o}$ that
stand for ``beta,'' ``substitution,'' and ``overheads'' respectively.

Let $\mathcal{L}$ be a fixed countable (infinite) set of \emph{names}.
The state of the transition system $s = (\mathbb{G},p,h,m)$ consists
of the following elements:
\begin{itemize}
 \item a \emph{named} well-boxed graph $\mathbb{G} = (G,\ell_G)$, that
       is a well-boxed graph $G$ with a \emph{naming} $l_G$ that
       assigns a unique name $\alpha \in \mathcal{L}$ to each node of
       $G$
 \item a pair $p = (e,d)$ called \emph{position}, of an edge $e$ of
       $G$ and a \emph{direction} $d \in \{ \uparrow,\downarrow\}$
 \item a \emph{history stack} $h$ defined by the grammar
       $
	h ::= \square
	\mid \mathsf{Ax}_\alpha:h \mid \mathsf{Cut}_\alpha:h
	\mid \otimes_\alpha:h \mid \parr_\alpha:h 
	\mid \oc_\alpha:h \mid \mathsf{D}_\alpha:h
	\mid \mathsf{C}^n_\alpha:h
       $,
       where $\alpha \in \mathcal{L}$ and $n$ is some positive natural
       number.
 \item a \emph{multiplicative stack} $m$ defined by the BNF grammar
       $m ::= \square \mid \mathsf{l}:m \mid \mathsf{r}:m$.
\end{itemize}
We refer to a node by its name, i.e.\ we say ``a node $\alpha$''
instead of ``a node whose name is $\alpha$.''

A pass transition
$(\mathbb{G},p,h,m) \dashrightarrow_\mathsf{o} (\mathbb{G},p',h',m')$
changes a
position using a multiplicative stack, pushes to a history stack, and
keeps a named graph unchanged.
All pass transitions have the label~$\mathsf{o}$.

Fig.~\ref{fig:PassTransitions} shows pass transitions graphically,
omitting irrelevant parts of graphs.
A position $p=(e,d)$ is represented by a bullet $\bullet$ (called
``token'') on the edge $e$ together with the direction $d$.
Recall that an edge with a strike-out represents a bunch of edges.
The transition in the last line of Fig.~\ref{fig:PassTransitions}
(where we assume $n > 0$)
moves a token from one of the incoming edges of a $\mathsf{C}_n$-node
to the outgoing edge of the node.
Node names $\alpha \in \mathcal{L}$ are indicated wherever needed.
\begin{figure}[htp]
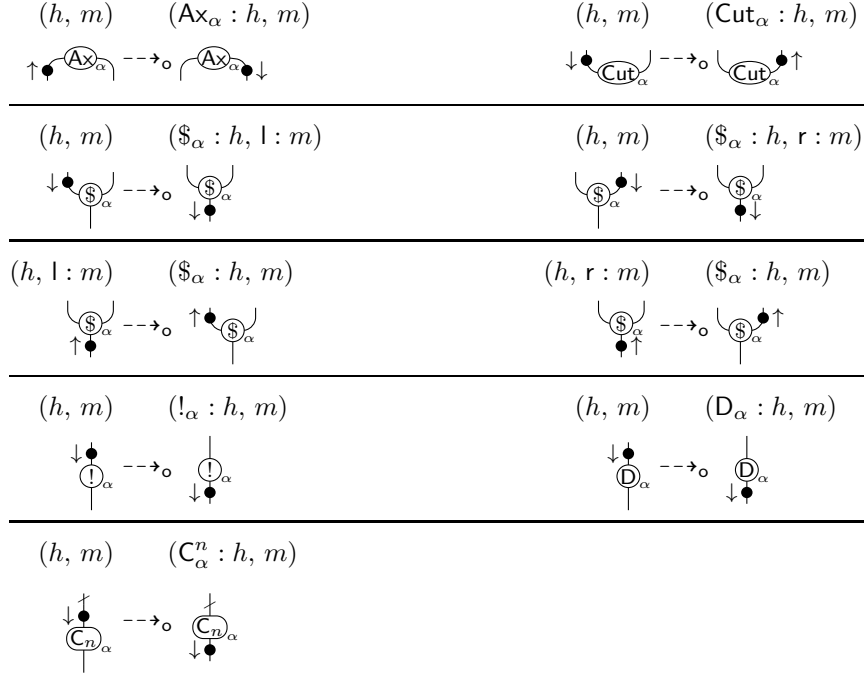

 
 \begin{align*}
  (h,\,m) &\qquad (\mathsf{Ax}_\alpha:h,\,m) &
  (h,\,m) &\qquad (\mathsf{Cut}_\alpha:h,\,m) \\
  \figPassTransitionsAxL &\dashrightarrow_\mathsf{o}
  \figPassTransitionsAxR &
  \figPassTransitionsCutL &\dashrightarrow_\mathsf{o}
  \figPassTransitionsCutR \\ \cline{1-4}
  (h,\,m) &\qquad (\$_\alpha:h,\,\mathsf{l}:m) &
  (h,\,m) &\qquad (\$_\alpha:h,\,\mathsf{r}:m) \\
  \figPassTransitionsMultLeftDownL &\dashrightarrow_\mathsf{o}
  \figPassTransitionsMultLeftDownR &
  \figPassTransitionsMultRightDownL &\dashrightarrow_\mathsf{o}
  \figPassTransitionsMultRightDownR \\ \cline{1-4}
  (h,\,\mathsf{l}:m) &\qquad (\$_\alpha:h,\,m) &
  (h,\,\mathsf{r}:m) &\qquad (\$_\alpha:h,\,m) \\
  \figPassTransitionsMultLeftUpL &\dashrightarrow_\mathsf{o}
  \figPassTransitionsMultLeftUpR &
  \figPassTransitionsMultRightUpL &\dashrightarrow_\mathsf{o}
  \figPassTransitionsMultRightUpR \\ \cline{1-4}
  (h,\,m) &\qquad (\oc_\alpha:h,\,m) &
  (h,\,m) &\qquad (\mathsf{D}_\alpha:h,\,m) \\
  \figPassTransitionsPrincipalL &\dashrightarrow_\mathsf{o}
  \figPassTransitionsPrincipalR &
  \figPassTransitionsDerelictionL &\dashrightarrow_\mathsf{o}
  \figPassTransitionsDerelictionR \\ \cline{1-4}
  (h,\,m) &\qquad (\mathsf{C}^n_\alpha:h,\,m) \\
  \figPassTransitionsContractionL &\dashrightarrow_\mathsf{o}
  \figPassTransitionsContractionR
 \end{align*}
 \caption{Pass Transitions ($\$ \in \{ \otimes,\parr \}$, $n > 0$)}
 \label{fig:PassTransitions}
\end{figure}

A rewrite transition
$(\mathbb{G},(e,d),h,m) \rightsquigarrow_\mathsf{x}
(\mathbb{G}',(e',d),h',m)$
consumes some elements of a history stack, rewrites a sub-graph of a
named graph, and updates a position (or, more precisely, its edge).
The label $\mathsf{x}$ of a rewrite transition
$\rightsquigarrow_\mathsf{x}$ is either $\mathsf{b}$, $\mathsf{s}$ or
$\mathsf{o}$.
Fig.~\ref{fig:RewriteTransitions} shows rewrite transition in the same
manner as Fig.~\ref{fig:PassTransitions}.
Multiplicative stacks are not present in the figure since they are
irrelevant.
The $\sharp$-node represents some arbitrary node (incoming edges
omitted).
We can see that no rewrite transition breaks the well-boxed-ness of a
graph.
\begin{figure}[htp]
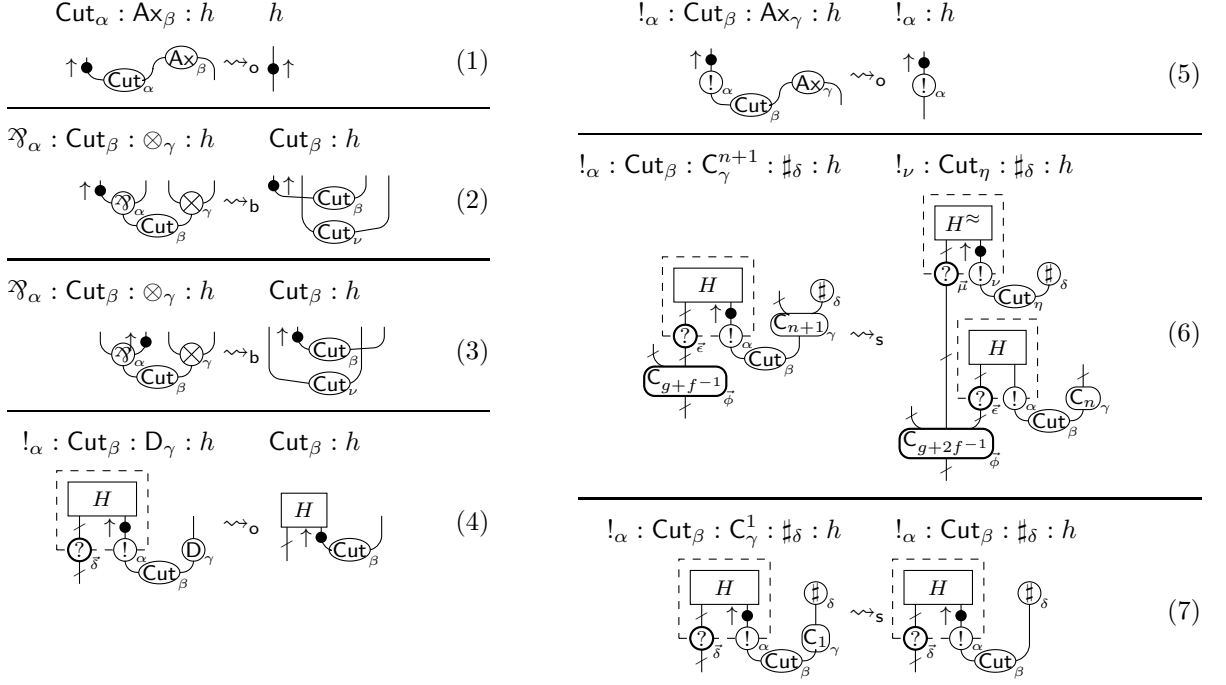

 \centering
 \centering
 \begin{minipage}[t]{.45\hsize}
  \begin{align}
   \mathsf{Cut}_\alpha:\mathsf{Ax}_\beta:h &\qquad h \notag \\
   \figRewriteTransitionsAxL
   &\rightsquigarrow_\mathsf{o}
   \figRewriteTransitionsAxR \label{RW:Ax1} \\ \cline{1-3}
   \parr_\alpha:\mathsf{Cut}_\beta:\otimes_\gamma:h
   &\qquad \mathsf{Cut}_\beta:h \notag \\
   \figRewriteTransitionsMultLeftL
   &\rightsquigarrow_\mathsf{b}
   \figRewriteTransitionsMultLeftR \label{RW:MultL} \\ \cline{1-3}
   \parr_\alpha:\mathsf{Cut}_\beta:\otimes_\gamma:h
   &\qquad \mathsf{Cut}_\beta:h \notag \\
   \figRewriteTransitionsMultRightL
   &\rightsquigarrow_\mathsf{b}
   \figRewriteTransitionsMultRightR \label{RW:MultR} \\ \cline{1-3}
   \oc_\alpha:\mathsf{Cut}_\beta:\mathsf{D}_\gamma:h
   &\qquad \mathsf{Cut}_\beta:h \notag \\
   \figRewriteTransitionsDerelictionL
   &\rightsquigarrow_\mathsf{o}
   \figRewriteTransitionsDerelictionR \label{RW:D} 
  \end{align}
 \end{minipage}%
 \begin{minipage}[t]{.55\hsize}
  \begin{align}
   \oc_\alpha:\mathsf{Cut}_\beta:\mathsf{Ax}_\gamma:h
   &\qquad \oc_\alpha:h \notag \\
   \figRewriteTransitionsAxPrincipalL
   &\rightsquigarrow_\mathsf{o}
   \figRewriteTransitionsAxPrincipalR \label{RW:Ax2} \\ \cline{1-3}
   \oc_\alpha:\mathsf{Cut}_\beta:\mathsf{C}^{n+1}_\gamma
   :\sharp_\delta:h
   &\qquad \oc_\nu:\mathsf{Cut}_\eta:\sharp_\delta:h \notag \\
   \figRewriteTransitionsContractionPosL
   &\rightsquigarrow_\mathsf{s}
   \figRewriteTransitionsContractionPosR \label{RW:CPos} \\ \cline{1-3}
   \oc_\alpha:\mathsf{Cut}_\beta:\mathsf{C}^1_\gamma
   :\sharp_\delta:h
   &\qquad \oc_\alpha:\mathsf{Cut}_\beta:\sharp_\delta:h \notag \\
   \figRewriteTransitionsContractionNullL
   &\rightsquigarrow_\mathsf{s}
   \figRewriteTransitionsContractionNullR \label{RW:COne}   
  \end{align}
 \end{minipage}
 \caption{Rewrite Transitions ($n > 0$)}
 \label{fig:RewriteTransitions}
\end{figure}

The rewrite transitions
(\ref{RW:Ax1}),(\ref{RW:MultL}),(\ref{RW:MultR}), and (\ref{RW:D}) are
exactly taken from MELL cut elimination \cite{Girard87LL}.
The rewrite transition (\ref{RW:Ax2}) is a variant of (\ref{RW:Ax1}).
It acts on a connected pair of a $\mathsf{Cut}$-node and an
$\mathsf{Ax}$-node that arises as a result of the transition
(\ref{RW:CPos}) or (\ref{RW:COne}) but cannot be rewritten by the
transition (\ref{RW:Ax1}).
These transitions (\ref{RW:CPos}) and (\ref{RW:COne}) are inspired by
the MELL cut elimination process for (binary) contraction nodes; note
that we assume $n > 0$ in Fig.~\ref{fig:RewriteTransitions}.

The rewrite transition (\ref{RW:CPos}) in
Fig.~\ref{fig:RewriteTransitions} deserves further explanation.
The sub-graph $H^\approx$ is a copy of the $\oc$-box $H$ where all the
names are replaced with fresh ones.
The thick $\mathsf{C}_{g+f^{-1}}$-node and
$\mathsf{C}_{g+2f^{-1}}$-node represent families
$\{ \mathsf{C}_{g(j)+f^{-1}(j)} \}_{j=0}^{m},
\{ \mathsf{C}_{g(j)+2f^{-1}(j)} \}_{j=0}^{m},$
of $\mathsf{C}$-nodes respectively.
They are connected to $\wn$-nodes
$\vec{\epsilon}=\epsilon_0,\ldots,\epsilon_l$ and
$\vec{\mu}=\mu_0,\ldots,\mu_l$
in such a way that:
\begin{itemize}
 \item the natural numbers $l,m$ satisfy $l \geq m$, and come with
       a surjection
       $f \colon \{ 0,\ldots,l \} \twoheadrightarrow \{ 0,\ldots,m \}$
       and a function $g \colon \{ 0,\ldots,m \} \to \N$ to the set
       $\N$ of natural numbers
 \item each $\wn$-node $\epsilon_i$ and each $\wn$-node $\mu_i$ are
       both connected to the $\mathsf{C}$-node $\phi_{f(i)}$
 \item each $\mathsf{C}$-node $\phi_j$ has $g(j)$ incoming
       edges whose source is none of the $\wn$-nodes
       $\vec{\epsilon},\vec{\mu}$.
\end{itemize}

Some rewrite transitions introduce new nodes to a graph.
We \emph{require} that the uniqueness of names throughout a whole
graph is not violated by these transitions.
Under this requirement, the introduced names $\nu,\vec{\mu}$ and the
renaming $H^\approx$ in Fig.~\ref{fig:RewriteTransitions} can be
arbitrary.
\begin{definition}
We call a state $((G,\ell_G),p,h,m)$ \emph{rooted at $e_0$} for an
open (outgoing) edge $e_0$ of $G$, if there exists a finite
sequence
$((G,\ell_G),(e_0,\uparrow),\square,\square) \dashrightarrow^*
((G,\ell_G),p,h,m)$
of pass transitions such that the position $p$ appears only last in
the sequence.
\end{definition}
Lem.~\ref{lem:DGoIMInvariants}(1) below implies that, the DGoIM can
determine whether a rewrite transition is possible at a rooted state by
only examining a history stack.
The rooted property is preserved by transitions.
\begin{lemma}[rooted states]
 \label{lem:DGoIMInvariants}
 Let $((G,\ell_G),(e,d),h,m)$ be a rooted state at $e_0$ with a
 (finite) sequence
 \[
 ((G,\ell_G),(e_0,{\uparrow}),\square,\square) \dashrightarrow^*
 ((G,\ell_G),(e,d),h,m).
 \]
 \begin{enumerate}
  \item History stack represents an (undirected and possibly cyclic)
	path of the graph $G$ connecting edges $e_0$ and~$e$.
  \item If a transition $((G,\ell_G),(e,d),h,m)
	\mathrel{(\dashrightarrow \cup \rightsquigarrow)}
	((G',\ell_{G'}),p',h',m')$
	is possible, the open edges of $G'$ are bijective to those of
	$G$, and the state $((G',\ell_{G'}),p',h',m')$ is rooted at
	the open edge corresponding to~$e_0$.
 \end{enumerate}
\end{lemma}
\begin{proof}
 (Sketch.)
 The proof of the first part is by induction on the length of the
 sequence of move transitions.
 For the second part,
 rewrite transitions $\rightsquigarrow$ modify open edges of a graph
 in a bijective way.
 The edge that a state is rooted at can be modified only by the
 rewrite transitions (\ref{RW:Ax1}) and (\ref{RW:Ax2}) involving
 $\mathsf{Ax}$-nodes.
\end{proof}

\subsection{Cost Analysis of the DGoIM}\label{sec:TimeCost}

The time cost of updating stacks is constant, as each transition changes
only a fixed number of top elements of stacks.
Updating a position is local and needs constant time, as it does
not require searching beyond the next edge in the graph from the current edge.
We can conclude all pass transitions take constant time.

We estimate the time cost of rewrite transitions by counting updated
nodes.
The rewrite transitions (\ref{RW:Ax1})--(\ref{RW:MultR}) involve a
fixed number of nodes, and the transition (\ref{RW:COne}) eliminates
one $\mathsf{C}_1$-node.
Only the transitions (\ref{RW:D}) and (\ref{RW:CPos}) have
non-constant time cost.
The number of doors deleted in the transition (\ref{RW:D}) can be
arbitrary, and so is the number of nodes introduced in the transition
(\ref{RW:CPos}).


Pass transitions and rewrite transitions are separately deterministic
(up to the choice of new names).
However, both a pass transition and a rewrite transition are possible
at some states.
We here opt for the following ``rewrites-first'' way to interleave
pass transitions with as much rewrite transitions as possible:
\begin{align*}
 s \rightarrowtriangle_\mathsf{x} s' \iffdef
 \begin{cases}
  s \rightsquigarrow_\mathsf{x} s'
  &\text{(if $\rightsquigarrow_\mathsf{x}$ possible)} \\
  s \dashrightarrow_\mathsf{x} s'
  &\text{(if only $\dashrightarrow[_\mathsf{x}$ possible)}.
 \end{cases}
\end{align*}
The DGoIM with this strategy yields a deterministic labelled
transition system $\rightarrowtriangle$ up to the choice of new names
in rewrite transitions.
We denote it by DGoIM{\rewritesfirst}, making the strategy explicit.
Note that there can be other strategies of interleaving although we do
not explore them here.

Before we conclude, several considerations about space cost analysis. Space costs are generally bound by time costs, so from our analysis there is an implicit guarantees that space usage will not explode. But if a more refined space cost analysis is desired, the following might prove to be useful. 

The space required in implementing a named well-boxed graph is bounded by
the number of its nodes.
The number of edges is linear in the number of nodes, because each
generator has a fixed out-degree and every edge of a well-boxed graph
has its source.

Additionally a $\oc$-box can be represented by associating its
auxiliary doors to its principal door.
This adds connections between doors to a graph that are as many as
$\wn$-nodes.
It enables the DGoIM to identify nodes of a $\oc$-box by following
edges from its principal and auxiliary doors.
Nodes in a $\oc$-box that are not connected to doors can be ignored,
since these nodes are never visited by a token (i.e.\ pointed by a
position) as long as the DGoIM acts on rooted states.

Only the rewrite transition (\ref{RW:CPos}) can increase the number of
nodes of a graph by copying a $\oc$-box with its doors.
Rewrite transitions can copy $\oc$-boxes and eliminate the $\oc$-box
structure, but they never create new $\oc$-boxes or change existing
ones.
This means that, in a sequence of transitions that starts with a graph
$G$, any $\oc$-boxes copied by the rewrite transition (\ref{RW:CPos})
are sub-graphs of the graph $G$.
Therefore the number of nodes of a graph increases linearly in the
number of transitions.

Elements of history stacks and multiplicative stacks, as well as a
position, are essentially pointers to nodes.
Because each pass/rewrite transition adds at most one element to each
stack, the lengths of stacks also grow linearly in the number of
transitions.

\section{Weak Simulation of the Call-by-Need SAM}

\subsection{Storeless Abstract Machine (SAM)}

We show the DGoIM{\rewritesfirst} implements call-by-need
evaluation by building a weak simulation of the call-by-need
Storeless Abstract Machine (SAM) defined in
Fig.~\ref{fig:StorelessAbstractMachine}.
It simplifies Danvy and Zerny's storeless machine
\cite[Fig.~8]{DanvyZ13} and accommodates a partial mechanism of
garbage collection (namely, transition (\ref{SAM:SOne})).
We will return to a discussion of garbage collection at the end of this section. 

The SAM is a labelled transition system between \emph{configurations}
$(\overline{t},E)$.
They are classified into two groups, namely \emph{term}
configurations and \emph{context} configurations, that are indicated
by annotations $\mathit{term},\mathit{ctxt}$ respectively.
Pure terms (resp.\ pure values) are terms (resp.\ values) that contain
no explicit substitutions $t[x \leftarrow u]$; we sometimes omit the
word ``pure'' and the overline in denotation as long as that raises no
confusion.
\begin{figure*}[t]
 \centering
 \begin{tabular}{rlrl}
  Terms
  & $t ::= x \mid \abs{x}{}{t} \mid \app{t}{t} \mid t[x \leftarrow t]$
  & \qquad Pure terms
  & $\overline{t} ::= x \mid \abs{x}{}{\overline{t}}
  \mid \app{\overline{t}}{\overline{t}}$ \\
  Values & $v ::= \abs{x}{}{t}$
  & Pure values & $\overline{v} ::= \abs{x}{}{\overline{t}}$ \\
  Evaluation contexts
  & $E ::= \emptyCtxt \mid \app{E}{\overline{t}}
  \mid E[x \leftarrow \overline{t}]
  \mid \plug{E}{x}[x \leftarrow E]$ \\
  Substitution contexts
  & $A ::= \emptyCtxt \mid A[x \leftarrow \overline{t}]$ \\
 \end{tabular}
 \begin{align}
  (\app{\overline{t}}{\overline{u}},\,E)_\mathit{term}
  &\to_\mathsf{o}
  (\overline{t},\,
  \plug{E}{\app{\emptyCtxt}{\overline{u}}})_\mathit{term}
  \label{SAM:O1} \\
  (x,\,\plug{E_1}{E_2[x \leftarrow \overline{t}]})_\mathit{term}
  &\to_\mathsf{o}
  (\overline{t},\,
  \plug{E_1}{\plug{E_2}{x}[x \leftarrow \emptyCtxt]})_\mathit{term}
  \label{SAM:O2} \\
  (\overline{v},\,E)_\mathit{term} &\to_\mathsf{o}
  (\overline{v},\,E)_\mathit{ctxt} \label{SAM:O3} \\
  (\abs{x}{}{\overline{t}},\,
  \plug{E}{\app{A}{\overline{u}}})_\mathit{ctxt}
  &\to_\mathsf{b}
  (\overline{t},\,
  \plug{E}{\plug{A}{
  \emptyCtxt[x \leftarrow \overline{u}]}})_\mathit{term}
  \label{SAM:B} \\
  (\overline{v},\,
  \plug{E_1}{\plug{E_2}{x}[x \leftarrow A]})_\mathit{ctxt}
  &\to_\mathsf{s}
  (\overline{v}^\approx,\,
  \plug{E_1}{\plug{A}{E_2[x \leftarrow \overline{v}]}})_\mathit{ctxt}
  & \text{(if $x \in \FV_\emptyset(E_2)$)} \label{SAM:SPos} \\
  (\overline{v},\,
  \plug{E_1}{\plug{E_2}{x}[x \leftarrow A]})_\mathit{ctxt}
  &\to_\mathsf{s}
  (\overline{v},\,\plug{E_1}{\plug{A}{E_2}})_\mathit{ctxt}
  & \text{(if $x \notin \FV_\emptyset(E_2)$)} \label{SAM:SOne}
 \end{align}
 \caption{Call-by-need Storeless Abstract Machine (SAM)}
 \label{fig:StorelessAbstractMachine}
\end{figure*}

Each evaluation context $E$ contains exactly one open hole
$\emptyCtxt$, and replacing it with a term $t$ (or an evaluation
context $E'$) yields a term $\plug{E}{t}$ (or an evaluation context
$\plug{E}{E'}$) called \emph{plugging}.
In particular an evaluation context
$\plug{E'}{x}[x \leftarrow E]$ replaces the open hole of $E'$ with $x$
and keeps the open hole of $E$.

Labels of transitions are the same as those used for the DGoIM
(i.e.\ $\mathsf{b}$, $\mathsf{s}$ and $\mathsf{o}$).
The transition (\ref{SAM:B}), with the label $\mathsf{b}$, corresponds
to the $\beta$-reduction
where evaluation and substitution of function arguments are delayed.
Substitution happens in the transitions (\ref{SAM:SPos}) and
(\ref{SAM:SOne}), with the label $\mathsf{s}$, that replaces
exactly one occurrence of a variable.
The other transitions with the label $\mathsf{o}$, namely
$(\overline{t},E) \to_\mathsf{o} (\overline{t'},E')$,
search a redex by rearranging a configuration.
The two pluggings $\plug{E}{\overline{t}}$ and
$\plug{E'}{\overline{t'}}$ indeed yield exactly the same term.

We characterise ``free'' variables using multisets of variables.
Multisets make explicit how many times a variable is duplicated in a
term (or an evaluation context).
This information of duplication is later used in
translating terms to graphs.
\begin{notation*}[multiset]
 A multiset $\mathbf{x} := [x,\ldots,x]$ consists of a finite number
 of $x$.
 The multiplicity of $x$ in a multiset $M$ is denoted by $M(x)$.
 We write $x \in^k M$ if $M(x) = k$, $x \in M$ if $M(x) > 0$ and
 $x \notin M$ if $M(x) = 0$.
 A multiset $M$ comes with its support set $\mathrm{supp}(M)$.
 For two multisets $M$ and $M'$, their sum and difference are denoted
 by $M + M'$ and $M - M'$ respectively.
 Removing \emph{all} $x$ from a multiset $M$ yields the multiset
 $M \backslash x$, e.g.\ ${[x,x,y] \backslash x} = [y]$.
\end{notation*}
\noindent
Each term $t$ and each evaluation context $E$ are respectively
assigned multisets of variables $\FV(t),\FV_M(E)$, with $M$ a multiset
of variables.
The multisets $\FV$ are defined inductively as follows.
\begin{align*}
 \FV(x) &:= [x], \\
 \FV(\abs{x}{}{t}) &:= \FV(t) \backslash x, \\
 \FV(\app{t}{u}) &:= \FV(t) + \FV(u), \\
 \FV(t[x \leftarrow u]) &:= (\FV(t) \backslash x) + \FV(u). \\
 \FV_M(\emptyCtxt) &:= M, \\
 \FV_M(\app{E}{t}) &:= \FV_M(E) + \FV(t), \\
 \FV_M(E[x \leftarrow t]) &:=
 (\FV_M(E)) \backslash x + \FV(t), \\
 \FV_M(\plug{E'}{x}[x \leftarrow E]) &:=
 (\FV_{[x]}(E')) \backslash x + \FV_M(E).
\end{align*}
The following equations can be proved by a straightforward induction
on $E$.
\begin{lemma}[decomposition]
 \label{lem:FVComposition}
 \begin{align*}
  \FV(\plug{E}{t}) &= \FV_{\FV(t)}(E) \\
  \FV_M(\plug{E}{E'}) &= \FV_{\FV_M(E')}(E)
 \end{align*}
\end{lemma}

A variable $x$ is \emph{bound} in a term $t$ if it appears in the form
of $\abs{x}{}{u}$ or $u[x \to u']$.
A variable $x$ is \emph{captured} in an evaluation context $E$ if it
appears in the form of $E'[x \leftarrow \overline{t}]$ (but not in the
form of $\plug{E'}{x}[x \leftarrow E'']$).
The transitions (\ref{SAM:SPos}) and (\ref{SAM:SOne}) depend on
whether or not the bound variable $x$ appears in the evaluation
context $E_2$.
If the variable $x$ appears, the value $\overline{v}$ is kept
for later use and its copy $\overline{v}^\approx$ is substituted for
$x$.
If not, the value $\overline{v}$ itself is substituted for $x$.

The SAM does not assume the $\alpha$-equivalence, but explicitly deals
with it in copying a value.
The copy $\overline{v}^\approx$ in has all its bound variables
replaced by distinct fresh variables (i.e.\ distinct variables that do
not appear in a whole configuration).
This implies that the SAM is deterministic up to the choice of new
variables introduced in copying.

A term $t$ is \emph{closed} if $\FV(t) = \emptyset$; and is
\emph{well-named} if each variable gets bound at most once in $t$,
and each bound variable $x$ in $t$ satisfies $x \notin \FV(t)$.
An \emph{initial} configuration is a term configuration
$(\overline{t_0},\emptyCtxt)_\mathit{term}$ where $\overline{t_0}$ is
closed and well-named.
A finite sequence of transitions from an initial configuration is
called an \emph{execution}.
A \emph{reachable} configuration $(\overline{t},E)$, that is a
configuration coming with an execution
from some initial configuration to itself, satisfies the following
invariant properties.
\begin{lemma}[reachable configurations]
 \label{lem:SAMInvariants}
 Let $(\overline{t},E)$ be a reachable configuration from an initial
 configuration $(\overline{t_0},\emptyCtxt)_\mathit{term}$.
 The term $\overline{t}$ is a sub-term of the initial term
 $\overline{t_0}$ up to $\alpha$-equivalence, and the plugging
 $\plug{E}{\overline{t}}$ is closed and well-named.
\end{lemma}
\begin{proof}
 (Sketch.)
 The proof is by induction on the length of the execution
 $(\overline{t_0},\emptyCtxt)_\mathit{term} \to^* (\overline{t},E)$.
 Not only the term $\overline{t}$ but also the term $u'$ in any
 sub-term $\app{t'}{u'}$ or $t'[x \leftarrow u']$ of the plugging
 $\plug{E}{\overline{t}}$ is a sub-term of the initial term
 $\overline{t_0}$.
 The transition (\ref{SAM:SPos}) renames a value in the way that
 preserves closedness and the well-named-ness of pluggings.
 In the transition (\ref{SAM:SOne}) where an explicit substitution for
 a bound variable $x$ is eliminated, the induction hypothesis ensures
 that the variable $x$ does not occur in the plugging
 $\plug{E_1}{\plug{A}{\plug{E_2}{\overline{v}}}}$.
\end{proof}

We now conclude with a brief consideration on \textit{garbage collection.}
 Transition (\ref{SAM:SOne}) eliminates an explicit substitution
 and therefore implements a partial mechanism of garbage collection.
 The mechanism is partial because only an explicit substitution that
 is looked up in an execution can be eliminated, as illustrated below.
 The explicit substitution $[x \leftarrow \abs{z}{}{z}]$ is eliminated
 in the first example, but not in the second example because the bound
 variable $x$ does not occur.
 \begin{align*}
  (\app{(\abs{x}{}{x})}{(\abs{z}{}{z})},\emptyCtxt)_\mathit{term}
  &\to^* (\abs{z}{}{z},\emptyCtxt)_\mathit{ctxt} \\
  (\app{(\abs{x}{}{\abs{y}{}{y}})}{(\abs{z}{}{z})},
  \emptyCtxt)_\mathit{term}
  &\to^* (\abs{y}{}{y},
  \emptyCtxt[x \leftarrow \abs{z}{}{z}])_\mathit{ctxt}
 \end{align*}
 We incorporate this partial garbage collection to make clear the
 behaviour of the DGoIM{\rewritesfirst}, in particular the use of the
 rewrite transitions (\ref{RW:CPos}) and (\ref{RW:COne}).

\subsection{Translation and Weak Simulation}

A weak simulation is built on top of translations of terms and
evaluation contexts.
The translations $(\cdot)^\dag$ are inductively defined in
Fig.~\ref{fig:TranslationTerms} and
Fig.~\ref{fig:TranslationEvalCtxts}.
What underlies them is the so-called
``call-by-value'' translation of intuitionistic logic to linear logic.
This translates all and only values to $\oc$-boxes that can be copied
by rewrite transitions.
\begin{figure}
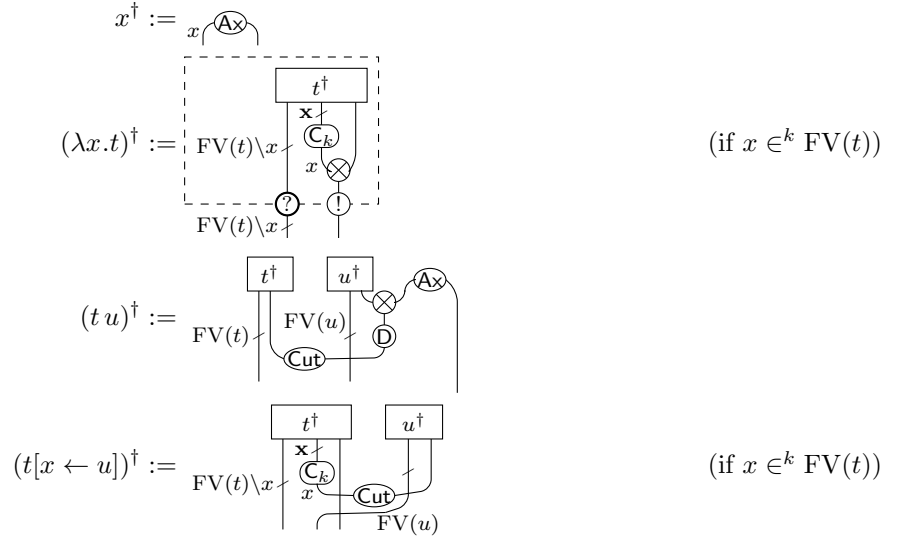

 \centering
 \begin{align*}
  x^\dag &:= \figTranslationTermsVar \\
  (\abs{x}{}{t})^\dag &:=
  \figTranslationTermsAbsCont \tag{if $x \in^k \FV(t)$} \\
  (\app{t}{u})^\dag &:= \figTranslationTermsApp \\
  (t[x \leftarrow u])^\dag &:=
  \figTranslationTermsESCont \tag{if $x \in^k \FV(t)$} \\
 \end{align*}
 \caption{Inductive Translation $(\cdot)^\dag$ of Terms to Well-boxed
 Graphs}
 \label{fig:TranslationTerms}
\end{figure}
\begin{figure}[ht]
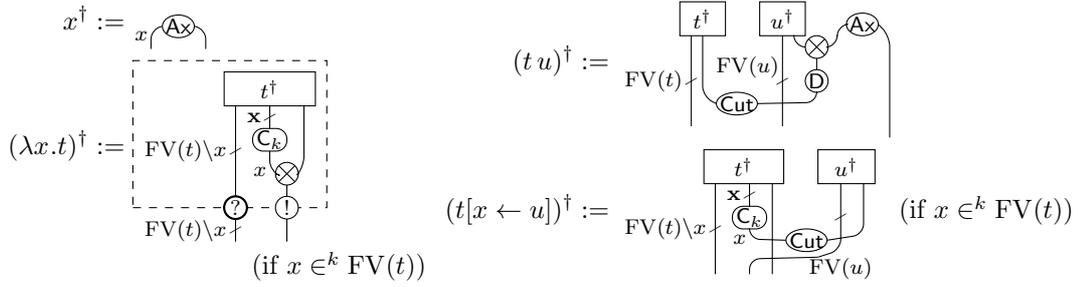

 \centering
 \begin{minipage}[t]{.4\hsize}
  \begin{align*}
   x^\dag &:= \figTranslationTermsVar \\
   (\abs{x}{}{t})^\dag &:=
   \figTranslationTermsAbsCont \tag{if $x \in^k \FV(t)$}
  \end{align*}  
 \end{minipage}%
 \begin{minipage}[t]{.5\hsize}
  \begin{align*}
   (\app{t}{u})^\dag &:= \figTranslationTermsApp \\
   (t[x \leftarrow u])^\dag &:=
   \figTranslationTermsESCont \tag{if $x \in^k \FV(t)$}
  \end{align*}  
 \end{minipage}
 \caption{Inductive Translation $(\cdot)_M^\dag$ of Evaluation Contexts to Graphs}
 \label{fig:TranslationEvalCtxts}
\end{figure}

The translation $\defTranslationsTerm$ of a term $t$ is a well-boxed
graph, where some edges are annotated with variables to help
understanding.
We continue representing a bunch of edges by a single edge and a
strike-out, with annotations denoted by a multiset, and a bunch of
nodes by a single thick node.
The translation $\defTranslationsEvalCtxt$ of an evaluation context
$E$, given a multiset $M$ of variables, is not a well-boxed graph
because it has incoming edges.
Lem.~\ref{lem:TranslationComposition} is analogous to
Lem.~\ref{lem:FVComposition}; their proof is by straightforward
induction on $E$.
\begin{lemma}[decomposition]
 \label{lem:TranslationComposition}
 \begin{align*}
  (\plug{E}{t})^\dag &= \lemTranslationCompositionTermConnected \\
  (\plug{E}{E'})_M^\dag &= \lemTranslationCompositionEvalCtxtConnected
 \end{align*}
\end{lemma}

The translations $(\cdot)^\dag$ are lifted to a binary relation
between reachable configurations of the SAM and rooted states of the
DGoIM{\rewritesfirst}.
\begin{definition}[binary relation $\preceq$]
 A reachable configuration $c$ and a state $((G,\ell_G),p,h,m)$
 satisfies $c \preceq ((G,\ell_G),p,h,m)$ if and only if:
 \begin{itemize}
  \item \begin{math}
	 (G,p) =
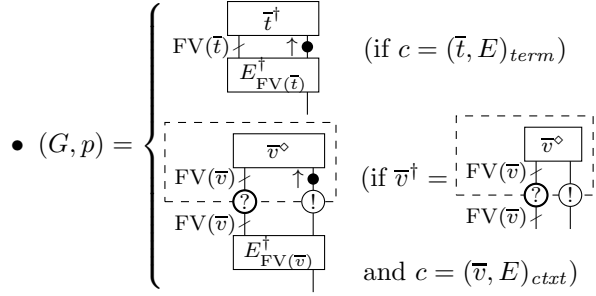
	 \begin{cases}
	  \defSimulationTermConfig
	  &\text{(if $c = (\overline{t},E)_\mathit{term}$)} \\
	  \defSimulationCtxtConfig
	  &\parbox[c]{4cm}{
	  (if $\overline{v}^\dag = \defSimulationValue$ \\
	  and $c = (\overline{v},E)_\mathit{ctxt}$)}
	 \end{cases}
	\end{math}
  \item $\ell_G$ is an arbitrary naming
  \item $((G,\ell_G),p,h,m)$ is rooted at the unique open edge of~$G$.
 \end{itemize}
\end{definition}
\noindent
Note that the graph $G$ in the above definition has exactly one open
edge, because it is equal to the translation
$\plug{E}{\overline{t}}^\dag$ (Lem.~\ref{lem:TranslationComposition})
and the plugging $\plug{E}{\overline{t}}$ is closed
(Lem.~\ref{lem:SAMInvariants}).

The binary relation $\preceq$ gives a weak simulation, as stated
below.
It is weak in Milner's sense \cite{Milner89}, where transitions with
the label $\mathsf{o}$ are regarded as internal.
We can conclude from Thm.~\ref{thm:Simulation} below that the
DGoIM{\rewritesfirst} soundly implements the call-by-need evaluation.
\begin{theorem}[weak simulation]
 \label{thm:Simulation}
 Let a configuration $c$ and a state $s$ satisfy $c \preceq s$.
 \begin{enumerate}
  \item If a transition $c \to_\mathsf{b} c'$ of the SAM is possible,
	there exists a sequence
	$s \rightarrowtriangle_\mathsf{o}^2
	\rightarrowtriangle_\mathsf{b}
	\rightarrowtriangle_\mathsf{o} s'$
	such that $c' \preceq s'$.
  \item If a transition $c \to_\mathsf{s} c'$ of the SAM is possible,
	there exists a sequence
	$s \rightarrowtriangle_\mathsf{s}
	\rightarrowtriangle_\mathsf{o} s'$
	such that $c' \preceq s'$.
  \item If a transition $c \to_\mathsf{o} c'$ of the SAM is possible,
	there exists a sequence
	$s \rightarrowtriangle_\mathsf{o}^N s'$
	such that $0 < N \leq 4$ and $c' \preceq s'$.
  \item No transition $\rightarrowtriangle$ is possible at the state
	$s'$ if $c' = (\overline{v},A)_\mathit{ctxt}$.
 \end{enumerate}
\end{theorem}
\begin{proof}
We start with two basic observations.
First, only a pass transition is possible at a
state $s$ that satisfies $(\overline{t},E)_\mathit{term} \preceq s$,
and second, no transition is possible at a state $s$ that satisfies
$(\overline{v},A)_\mathit{ctxt} \preceq s$.

Fig.~\ref{fig:SimulationIllustrated} shows how the
DGoIM{\rewritesfirst} simulates each transition of the SAM.
Stacks, annotations of edges, and some annotations of
$\mathsf{C}$-nodes are omitted.
The equations in Fig.~\ref{fig:SimulationIllustrated} apply the
decomposition properties in Lem.~\ref{lem:TranslationComposition} as
well as the other decomposition properties in the following
Lem.~\ref{lem:DecomposeTranslationEvalCtxt}.
In the application, we exploit the closedness and well-named-ness of
reachable configurations (in the sense of
Lem.~\ref{lem:SAMInvariants}).
\begin{figure*}[p]
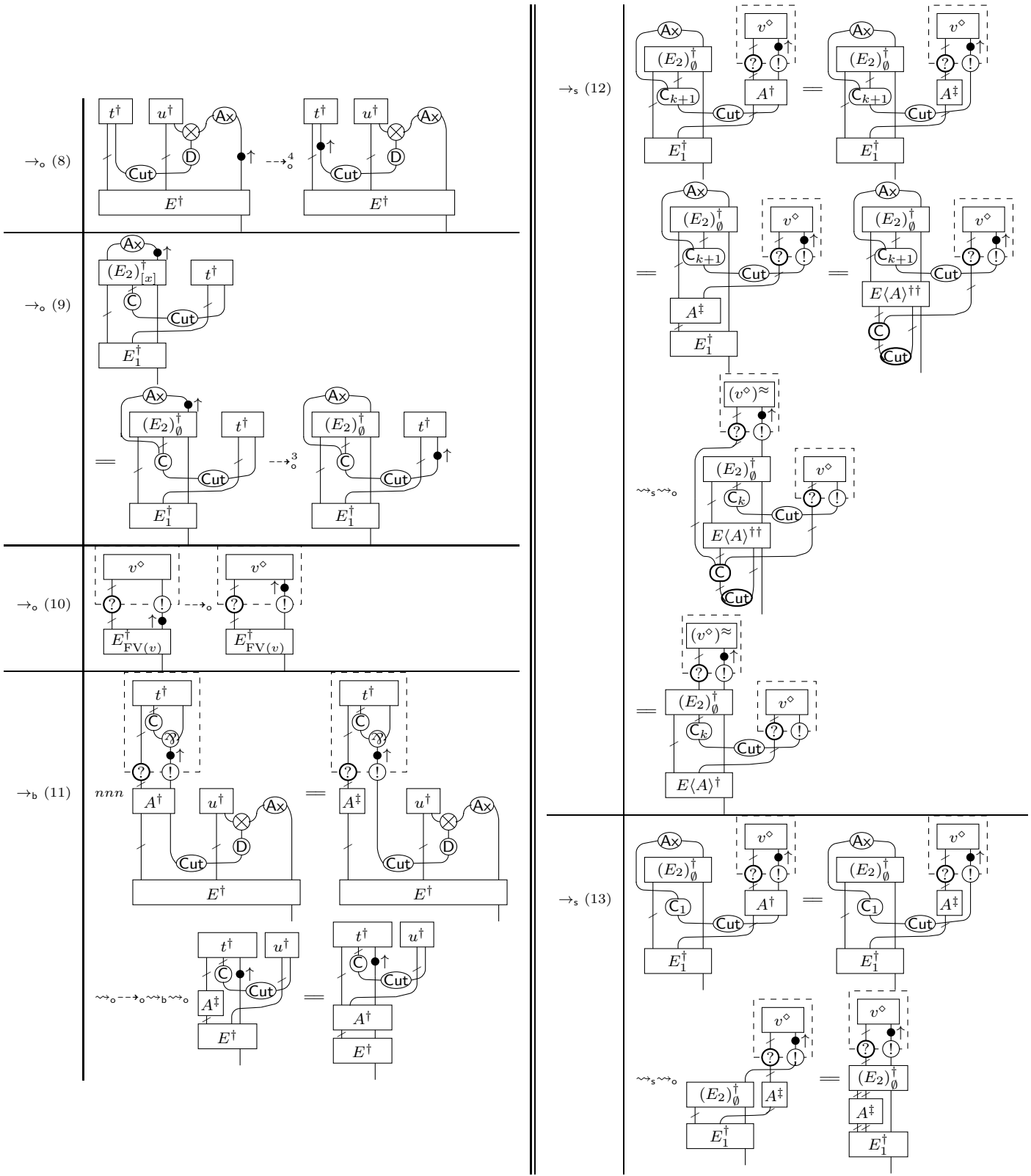

	\scriptsize
 \begin{tabular}{l||l}
  \begin{tabular}{r|l}
   $\to_\mathsf{o}$ (\ref{SAM:O1})
   & $\illustrateSimulationCOneFirst
   \dashrightarrow_\mathsf{o}^4
   \illustrateSimulationCOneSecond$ \\ \hline
   $\to_\mathsf{o}$ (\ref{SAM:O2})
   & $\illustrateSimulationCTwoFirst$ \\
   & $\longequal \illustrateSimulationCTwoSecond
   \dashrightarrow_\mathsf{o}^3
   \illustrateSimulationCTwoThird$ \\ \hline
   $\to_\mathsf{o}$ (\ref{SAM:O3})
   & $\illustrateSimulationCThreeFirst
   \dashrightarrow_\mathsf{o}
   \illustrateSimulationCThreeSecond$ \\ \hline
   $\to_\mathsf{b}$ (\ref{SAM:B})
   & $\illustrateSimulationMFirst
   \longequal \illustrateSimulationMSecond$ \\
   & $\rightsquigarrow_\mathsf{o}
   \dashrightarrow_\mathsf{o}
   \rightsquigarrow_\mathsf{b} \rightsquigarrow_\mathsf{o}
   \illustrateSimulationMThird
   \longequal \illustrateSimulationMFourth$
  \end{tabular}
  &
  \begin{tabular}{r|l}
   $\to_\mathsf{s}$ (\ref{SAM:SPos})
   & $\illustrateSimulationEOneFirst
   \longequal \illustrateSimulationEOneSecond$ \\
   & $\longequal \illustrateSimulationEOneThird
   \longequal \illustrateSimulationEOneFourth$ \\
   & $\rightsquigarrow_\mathsf{s} \rightsquigarrow_\mathsf{o}
   \illustrateSimulationEOneFifth$ \\
   & $\longequal \illustrateSimulationEOneSixth$ \\ \hline
   $\to_\mathsf{s}$ (\ref{SAM:SOne})
   & $\illustrateSimulationETwoFirst
   \longequal \illustrateSimulationETwoSecond$ \\
   & $\rightsquigarrow_\mathsf{s} \rightsquigarrow_\mathsf{o}
   \illustrateSimulationETwoThird
   \longequal \illustrateSimulationETwoFourth$
  \end{tabular}
 \end{tabular}
 \caption{Illustration of Simulation ($k > 0$)}
 \label{fig:SimulationIllustrated}
\end{figure*}
\end{proof}

\begin{lemma}[decomposition]
 \label{lem:DecomposeTranslationEvalCtxt}
 \noindent
 Let $M_0,M$ be multisets of variables.
 \begin{enumerate}
  \item The translation $A_M^\dag$ of a substitution context $A$ has a
	unique decomposition $\lemDecomposeTranslationSubstCtxt$.
  \item If no variables in $M_0$ are captured in an
	evaluation context $E$, the translation
	$E_{M_0 + M}^\dag$ is equal to the graph
	$\lemDecomposeTranslationEvalCtxtNotBound$.
  \item If each variable in $M_0$ is captured in an evaluation
	context $E$ exactly once, the translation
	$E_{M_0 + M}^\dag$ has a unique decomposition
	$\lemDecomposeTranslationEvalCtxtBound$.
	The multiset $M_1$ satisfies
	$\mathrm{supp}(M_1) \subseteq \mathrm{supp}(M_0)$, and the
	thick $\mathsf{C}_{M_0 + M_1}$-node represents a family
	$\{ \mathsf{C}_{M_0(x)+M_1(x)} \}_{x \in \mathrm{supp}(M_0)}$
	of $\mathsf{C}$-nodes.
 \end{enumerate}
\end{lemma}
\begin{proof}
 The proofs for 1.\ and 2.\ are by straightforward inductions on $A$
 and $E$ respectively.
 The proof for 3.\ is by induction on the dimension of $M_0$, i.e.\
 the size of the support set $\mathrm{supp}(M_0)$.
 The base case where $M_0 = \emptyset$ is obvious.
 In the inductive case, let $x$ satisfy $x \in M_0$.
 Since $x$ is captured exactly once in $E$ by assumption, the
 evaluation context $E$ can be decomposed as
 $E = \plug{E_1}{E_2[x \leftarrow t]}$ such that $x$ is not captured
 in $E_1$ or $E_2$.
 The evaluation context $E_2$ satisfies $x \in^k \FV_{M_0+M}(E_2)$ for
 some positive multiplicity $k$.
 Moreover the multiset $M_0$ can be decomposed as
 $M_0 = M_1 + \mathbf{x} + M_2$ such that all the variables in $M_1$
 (resp.\ $M_2$) are only captured in $E_1$ (resp.\ $E_2$).
 The translation $E_{M_0 + M}^\dag$ can be decomposed as in Fig.~\ref{fig:decomp}.
\begin{figure}
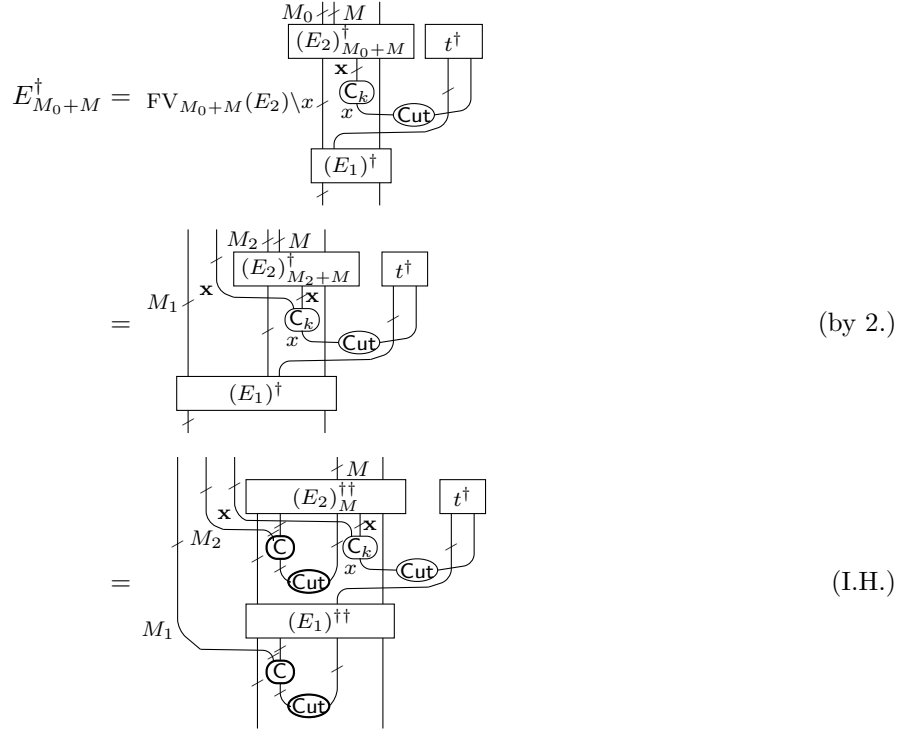

 
 \begin{align*}
  E_{M_0+M}^\dag &= \lemDecomposeTranslationEvalCtxtProofFirst \\
  &= \lemDecomposeTranslationEvalCtxtProofSecond \tag{by 2.} \\
  &= \lemDecomposeTranslationEvalCtxtProofThird \tag{I.H.}
 \end{align*}
\caption{Decomposing $E_{M_0 + M}^\dag$}
\label{fig:decomp}
\end{figure}
 Finally, we let $E_M^{\dag\dag}$ consist of
 $(E_2)_M^{\dag\dag},t^\dag,(E_1)^{\dag\dag}$.
\end{proof}


\section{Time Cost Analysis of Rewrites-First Interleaving}

\subsection{Recipe for Time Cost Analysis}

Our time cost analysis of the DGoIM{\rewritesfirst} follows
Accattoli's recipe, described in \cite{AccattoliBM14,Accattoli16}, of
analysing complexity of abstract machines.
This section recalls the recipe and explains how it applies to the
DGoIM{\rewritesfirst}.

The time cost analysis focuses on how efficiently an abstract machine
implements an evaluation strategy.
In other words, we are not interested in minimising the number of
$\beta$-reduction steps simulated by an abstract machine.
Our interest is in making the number of transitions of an
abstract machine ``reasonable,'' compared to the number of necessary
$\beta$-reduction steps determined by a given evaluation strategy.

Accattoli's recipe assumes that an abstract machine has three groups
of transitions: 1) ``$\beta$-transitions'' that correspond to
$\beta$-reduction in which substitution is delayed, 2) transitions
perform substitution, and 3) other ``overhead'' transitions.
We incorporate this classification using the labels
$\mathsf{b},\mathsf{s},\mathsf{o}$ of transitions.

Another assumption of the recipe is that, each step of
$\beta$-reduction is simulated by a single transition of an abstract
machine, and so is substitution of each occurrence of a variable.
This is satisfied by many known abstract machines including the SAM,
however not by the DGoIM{\rewritesfirst}.
The DGoIM{\rewritesfirst} has ``finer'' transitions and can take
several transitions to simulate a single step of reduction (hence a
single transition of the SAM, as we can observe in
Thm.~\ref{thm:Simulation}).
In spite of this mismatch we can still follow the recipe, thanks to
the weak simulation $\preceq$.
It discloses what transitions of the DGoIM exactly correspond to
$\beta$-reduction and substitution, and gives a concrete number of
overhead transitions that the DGoIM{\rewritesfirst} needs to simulate
$\beta$-reduction and substitution. The recipe for the time cost analysis is:
\begin{enumerate}
 \item Examine the number of transitions, by means of the size of
       input and the number of $\beta$-transitions.
 \item Estimate time cost of single transitions.
 \item Derive a bound of the overall execution time cost.
 \item Classify an abstract machine according to its execution time
       cost.
\end{enumerate}
The last step is accompanied by the following taxonomy of abstract
machines introduced in \cite{Accattoli16}.
\begin{definition}
 [classes of abstract machines {\cite[Def.~7.1]{Accattoli16}}]
 \label{def:taxonomy}
 \noindent
 \begin{enumerate}
  \item An abstract machine is \emph{efficient} if its execution time
	cost is linear in both the input size and the number of
	$\beta$-transitions.
  \item An abstract machine is \emph{reasonable} if its execution time
	cost is polynomial in the input size and the number of
	$\beta$-transitions.
  \item An abstract machine is \emph{unreasonable} if it is not
	reasonable.
 \end{enumerate}
\end{definition}

The input size in our case is given by the \emph{size} $|t|$ of a term
$t$, inductively defined by:
\begin{align*}
 |x| &:= 1
 & |\abs{x}{}{t}| &:= |t| + 1 \\
 |\app{t}{u}| &:= |t| + |u| + 1
 & |t[x \leftarrow u]| &:= |t| + |u| + 1.
\end{align*}
Given a sequence $r$ of transitions (of either the SAM or the
DGoIM{\rewritesfirst}), we denote the number of transitions with a
label $\mathsf{x}$ in $r$ by $|r|_\mathsf{x}$.
Since we use the fixed set $\{ \mathsf{b},\mathsf{s},\mathsf{o} \}$ of
labels, the length $|r|$ of the sequence $r$ is equal to the sum
$|r|_\mathsf{b} + |r|_\mathsf{s} + |r|_\mathsf{o}$.

\subsection{Number of Transitions}

We first estimate the number of transitions of the SAM, and then
derive estimation for the DGoIM{\rewritesfirst}.
\begin{lemma}[quantitative bounds for SAM]
 \label{lem:SAMBounds}
 Each execution $e$ from an initial configuration
 $(\overline{t_0},E)_\mathit{term}$, comes with the following
 inequalities:
 \begin{align*}
  |e|_\mathsf{s} &\leq |e|_\mathsf{b} \\
  |e|_\mathsf{o}
  &\leq |\overline{t_0}| \cdot (5 \cdot |e|_\mathsf{b} + 2)
  + (3 \cdot |e|_\mathsf{b} + 1).
 \end{align*}
\end{lemma}
\begin{proof}
 The proof is analogous to the discussion in
 \cite[Sec.~11]{AccattoliBM14}.
 Let $|e|^{(\ref{SAM:O1})}$,
 $|e|^{(\ref{SAM:O2})}$ and
 $|e|^{(\ref{SAM:O3})}$ be the numbers of transitions
 (\ref{SAM:O1}), (\ref{SAM:O2}) and (\ref{SAM:O3}) in $e$,
 respectively.
 The number $|e|_\mathsf{o}$ is equal to the sum of these three
 numbers.
 
 The transition (\ref{SAM:B}) introduces one explicit
 substitution, and the other transitions never increase the
 number of explicit substitution.
 In particular the transition (\ref{SAM:SPos}) copies a pure value,
 that means no explicit substitutions are copied.
 Therefore we can bound the number of transitions that concern
 explicit substitutions, and obtain
 $|e|^{(\ref{SAM:O2})} \leq |e|_\mathsf{b}$ and
 $|e|_\mathsf{s} \leq |e|_\mathsf{b}$.

 Each occurrence of the transition (\ref{SAM:O3}) in an execution is
 either the last transition of the execution or followed by
 the transitions (\ref{SAM:B}), (\ref{SAM:SPos}) and (\ref{SAM:SOne}).
 This yields the inequality
 $|e|^{(\ref{SAM:O3})} \leq |e|_\mathsf{b} + |e|_\mathsf{s} + 1$
 and hence
 $|e|^{(\ref{SAM:O3})} \leq 2 \cdot |e|_\mathsf{b} + 1$.

 The transition (\ref{SAM:O1}) reduces the size of a pure term that is
 the first component of a configuration.
 The pure term is always a sub-term of the initial term
 $\overline{t_0}$ (Lem.~\ref{lem:SAMInvariants}).
 This means each maximal subsequence of an execution that solely
 consists of the transition (\ref{SAM:O1}) has at most the length
 $|\overline{t_0}|$.
 Such a maximal subsequence either occurs at the end of an execution
 or is followed by transitions other than the transition
 (\ref{SAM:O1}).
 Therefore the number of these maximal sub-sequences is no more than
 $|e|_\mathsf{b} + |e|_\mathsf{s}
 + |e|^{(\ref{SAM:O2})} + |e|^{(\ref{SAM:O3})} + 1$
 that can be bounded by $5 \cdot |e|_\mathsf{b} + 2$.
 A bound of the number $|e|^{(\ref{SAM:O1})}$ can be given
 by multiplying these two bounds, namely we obtain
 $|e|^{(\ref{SAM:O1})} \leq
 |\overline{t_0}| \cdot (5 \cdot |e|_\mathsf{b} + 2)$.
\end{proof}

Combining these bounds for the SAM with the weak simulation $\preceq$,
we can estimate the number of transitions of the DGoIM{\rewritesfirst}
as below.
\begin{proposition}[quantitative bounds for DGoIM{\rewritesfirst}]
 \label{prop:DGoIMBounds}
 Let $r \colon s_0 \rightarrowtriangle^* s$ be a sequence of
 transitions of the DGoIM{\rewritesfirst}.
 If there exists an execution
 $(\overline{t_0},\emptyCtxt)_\mathit{term} \to^* (\overline{t},E)$ of
 the SAM such that
 $s_0 \preceq (\overline{t_0},\emptyCtxt)_\mathit{term}$ and
 $s \preceq (\overline{t},E)$,
 the sequence $r$ comes with the following inequalities:
 \begin{align*}
  |r|_\mathsf{s} &\leq |r|_\mathsf{b} \\
  |r|_\mathsf{o}
  &\leq 4 \cdot |\overline{t_0}| \cdot (5 \cdot |r|_\mathsf{b} + 2)
  + (16 \cdot |r|_\mathsf{b} + 4).
 \end{align*}
\end{proposition}
\begin{proof}
 This is a direct consequence of Lem.~\ref{lem:SAMBounds} and
 Thm.~\ref{thm:Simulation}.
\end{proof}

\subsection{Execution Time Cost}

We already discussed time cost of single transitions of the DGoIM in
Sec.~\ref{sec:TimeCost}.
It is worth noting that the discussion in Sec.~\ref{sec:TimeCost} is
independent of any particular choice of a rewriting and token-passing interleaving strategy. 

Thm.~\ref{thm:TimeCost} below gives a bound of execution time cost of
the DGoIM{\rewritesfirst}.
We can conclude that, according to Accattoli's taxonomy (see
Def.~\ref{def:taxonomy}), the DGoIM{\rewritesfirst} is ``efficient''
as an abstract machine for the call-by-need evaluation.
\begin{theorem}[time cost]
 \label{thm:TimeCost}
 Let $C,D$ be fixed natural numbers, and
 $r \colon s_0 \rightarrowtriangle^* s$ be a sequence of transitions
 of the DGoIM{\rewritesfirst}.
 If there exists an execution
 $(\overline{t_0},\emptyCtxt)_\mathit{term} \to^* (\overline{t},E)$ of
 the SAM such that
 $s_0 \preceq (\overline{t_0},\emptyCtxt)_\mathit{term}$ and
 $s \preceq (\overline{t},E)$,
 the total time cost $T(r)$ of the sequence $r$ satisfies:
 \begin{align*}
  T(r)
  &= \mathcal{O}((|\overline{t_0}|+C)\cdot(|r|_\mathsf{b}+D)).
 \end{align*}
\end{theorem}
\begin{proof}
 We estimated in Sec.~\ref{sec:TimeCost} that time cost of single
 transitions, except for the rewrite transitions (\ref{RW:D}) and
 (\ref{RW:CPos}), is constant.
 Time cost of these rewrite transitions (\ref{RW:D}) and
 (\ref{RW:CPos}) depends on the number of doors and/or nodes of a
 $\oc$-box.

 Since the sequence $r$ of transitions simulates an execution of the
 SAM, every $\oc$-box concerned in $r$ arises as the translation of a
 value $\overline{v}$.
 By definition of the translation $(\cdot)^\dag$
 (Fig.~\ref{fig:TranslationTerms}), the graph
 $\overline{v}^\dag$ is a $\oc$-box, with as many auxiliary doors as
 occurrences of free variables in $\overline{v}$.
 The number of auxiliary doors is no more than the number of nodes in
 the $\oc$-box, due to the well-boxed-ness condition, that is linear in
 the size $|\overline{v}|$ of the value.

 Moreover the value $\overline{v}$ appears as the first component of a
 context configuration, and therefore it is a sub-term of the initial
 term $\overline{t_0}$ (Lem.~\ref{lem:SAMInvariants}).
 As a result, time cost of each occurrence of the rewrite
 transitions (\ref{RW:D}) and (\ref{RW:CPos}) in the sequence $r$ is
 linear in the size $\overline{t_0}$.

 The bound of the total time cost $T(r)$ of the sequence $r$ is given
 by combining these estimations for single transitions with the
 results of Prop.~\ref{prop:DGoIMBounds}.
\end{proof}
\begin{corollary}
 The DGoIM{\rewritesfirst} is an efficient abstract machine, in the
 sense of Def.~\ref{def:taxonomy}.
\end{corollary}

\section{Conclusions}

We introduced the DGoIM, which can interleave token passing with graph
rewriting informed by the trajectory of the token.
We focused on the rewrites-first interleaving and proved that it
enables the DGoIM to implement the call-by-need evaluation strategy.
The quantitative analysis of time cost certified that the
DGoIM{\rewritesfirst} gives an ``efficient'' implementation in the
sense of Accattoli's classification.
The proof of Thm.~\ref{thm:TimeCost} pointed out that eliminating and
copying $\oc$-boxes are two main sources of time cost.
Our results are built on top of a weak simulation of the SAM, that
relates several transitions of the DGoIM to each computational task
such as $\beta$-reduction and substitution.

The main feature of the DGoIM is the flexible combination of
interaction and rewriting.
We here briefly discuss how the flexibility can enable the DGoIM to
implement evaluation strategies other than the call-by-need.

As mentioned in Sec.~\ref{sec:IntroductionInterleaving}, the
passes-only interleaving yields an ordinary token-passing abstract
machine that is known to implement the call-by-name evaluation.
Because no rewrites are triggered, as oppose to the rewrites-first
interleaving, a token not only can pass a principal door but also can
go inside a $\oc$-box and pass an auxiliary door.
These behaviours are in fact not possible with the DGoIM presented in
this paper; the transitions and data carried by a token are
tailored to the rewrites-first interleaving.
To recover an ordinary token-passing machine, we therefore need to add
pass transitions that involve auxiliary doors and data structures
(so-called ``exponential signatures'') that deal with $\oc$-boxes, for
example.

The only difference between the call-by-need and the call-by-value
evaluations lies in when function arguments are evaluated.
In the DGoIM, this corresponds to changing a trajectory of a token so
that it visits function arguments immediately after it detects
function application.
Therefore, to implement the call-by-value evaluation, the DGoIM can
still use the rewrites-first interleaving, but it should use a
modified set of pass transitions. Further refinements, not only of the evaluation strategies but also of the graph representation could yield even more efficient implementation, such as \textit{full lazy evaluation}, as hinted in~\cite{Sinot05}.

Our final remarks concern programming features that have been
modelled using token-passing abstract machines.
Ground-type constants are handled by attaching memories to either
nodes of a graph or a token, in e.g.\
\cite{Mackie95,HoshinoMH14,DalLagoFVY15} ---
this can be seen as a simple form of graph rewriting.
Algebraic effects are also accommodated using memories attached to
nodes of a graph in token machines~\cite{HoshinoMH14}, but their treatment would be much simplified in the DGoIM as effects are evaluated out of the term via rewriting. 

\subparagraph*{Acknowledgements.}

We are grateful to Ugo Dal Lago and anonymous reviewers for
encouraging and insightful comments on earlier versions of this
work.

\bibliographystyle{plain}

\end{document}

%% file: graph.tex
\newcommand{\putAnchor}[2]{ 
\coordinate (#2) at #1;
}

\newcommand{\midAnchor}[3]{	
\coordinate (#3) at ($(#1)!0.5!(#2)$);
}
\newcommand{\transAnchor}[3]{	
\coordinate (#3) at ($(#1) + #2$);
}

\newcommand{\outline}{}		
\newcommand{\thickened}{ \draw [thick] \outline; }

\newcommand{\drawNode}[6]{	
\renewcommand{\outline}{
#2 ellipse [x radius=#3, y radius=#4]
}
\filldraw [fill=white] \outline node {#1};
\path [xscale=#3, yscale=#4] (-30:1) coordinate (OR);
\node at ($#2 + (OR)$) [anchor=north west, inner sep=0] {\tiny #5};
\putAnchor{#2}{#6}
\putAnchor{($#2 + (0,#4)$)}{#6I}
\putAnchor{($#2 + (0,-#4)$)}{#6O}
\putAnchor{($#2 + (-#3,0)$)}{#6L}
\putAnchor{($#2 + (#3,0)$)}{#6R}
\putAnchor{($#2 + (-#3,0)$)}{#6IL}
\putAnchor{($#2 + (#3,0)$)}{#6IR}
}
\newcommand{\drawRecNode}[6]{	
\renewcommand{\outline}{
[rounded corners] ($#2 + (-#3,-#4)$) rectangle ($#2 + (#3,#4)$)
}
\filldraw [fill=white] \outline;
\node at #2 {#1};
\path [xscale=#3, yscale=#4] (-30:1) coordinate (OR);
\node at ($#2 + (OR)$) [anchor=north west, inner sep=0.5pt] {\tiny #5};
\putAnchor{#2}{#6}
\putAnchor{($#2 + (0,#4)$)}{#6I}
\putAnchor{($#2 + (0,-#4)$)}{#6O}
\putAnchor{($#2 + (-#3,0)$)}{#6L}
\putAnchor{($#2 + (#3,0)$)}{#6R}
\putAnchor{($#2 + 0.5*(-#3,0) + (0,#4)$)}{#6IL}
\putAnchor{($#2 + 0.5*(#3,0) + (0,#4)$)}{#6IR}
}

\newcommand{\drawDummyNode}[3]{
\drawNode{$\sharp$}{#1}{1}{1}{#2}{Dummy#3}
}
\newcommand{\drawAxNode}[3]{	
\drawNode{$\mathsf{Ax}$}{#1}{1.4}{1}{#2}{Ax#3}
}
\newcommand{\drawCutNode}[3]{	
\drawNode{$\mathsf{Cut}$}{#1}{1.8}{1}{#2}{Cut#3}
}
\newcommand{\drawTensorNode}[3]{	
\drawNode{}{#1}{1}{1}{#2}{Tensor#3}
\draw ($#1 + (45:1)$)--($#1 + (225:1)$);
\draw ($#1 + (-45:1)$)--($#1 + (-225:1)$);
}
\newcommand{\drawParrNode}[3]{	
\drawNode{$\parr$}{#1}{1}{1}{#2}{Parr#3}
}
\newcommand{\drawDollarNode}[3]{	
\drawNode{$\$$}{#1}{1}{1}{#2}{Dollar#3}
}
\newcommand{\drawPrincipalNode}[3]{	
\drawNode{$\oc$}{#1}{1}{1}{#2}{Principal#3}
}
\newcommand{\drawAuxiliaryNode}[3]{	
\drawNode{$\wn$}{#1}{1}{1}{#2}{Auxiliary#3}
}
\newcommand{\drawDerelictionNode}[3]{	
\drawNode{$\mathsf{D}$}{#1}{1}{1}{#2}{Dereliction#3}
}

\newcommand{\drawContractionNode}[6]{	
\drawRecNode{$\mathsf{C}_{#1}$}{#2}{#3}{#4}{#5}{Contraction#6}
}

\newcommand{\withUpToken}[1]{	
\fill (midpoint) circle [radius=0.5] node [#1] {$\uparrow$};
}
\newcommand{\withDownToken}[1]{	
\fill (midpoint) circle [radius=0.5] node [#1] {$\downarrow$};
}

\newcommand{\strikedOut}{
\node at (midpoint) [rotate=-60] {\tiny$\mid$};
}

\newcommand{\withLabel}[3]{	
\node [anchor=#1, inner sep=#2] at (label) {#3};
}
\newcommand{\near}{0.1em}
\newcommand{\far}{0.3em}

\newcommand{\drawCorneredConnection}[3]{ 
\draw [#1] let \p1 = ($(#3) - (#2)$) in
(#2)[rounded corners]--++(\x1,0.1*\y1) coordinate (label)
--(#3) coordinate [midway] (midpoint);
}

\newcommand{\drawCorneredOpen}[3]{ 
\draw [#1] let \p1 = #3 in
(#2)[rounded corners]--++(\x1,0.1*\y1) coordinate (label)
--($(#2) + #3$) coordinate [midway] (midpoint);
}

\newcommand{\drawStraightConnection}[3]{ 
\draw [#1] (#2)--(#3)
coordinate [midway] (label) coordinate [midway] (midpoint);
}

\newcommand{\drawStraightOpen}[3]{ 
\draw [#1] (#2)--($(#2) + #3$)
coordinate [midway] (label) coordinate [midway] (midpoint);
}

%% file: picture.tex
\newcommand{\figGeneratorsAx}{\ifshowFigGenerators
\begin{tikzpicture}
 \drawAxNode{(0,0)}{}{}
 \drawCorneredOpen{->}{AxL}{(-1,-2)}
 \drawCorneredOpen{->}{AxR}{(1,-2)}
 \node at (0,2) {};	
\end{tikzpicture}
\fi}
\newcommand{\figGeneratorsCut}{\ifshowFigGenerators
\begin{tikzpicture}
 \drawCutNode{(0,0)}{}{}
 \drawCorneredOpen{<-}{CutL}{(-1,2)}
 \drawCorneredOpen{<-}{CutR}{(1,2)}
 \node at (0,-2) {};	
\end{tikzpicture}
\fi}
\newcommand{\figGeneratorsTensor}{\ifshowFigGenerators
\begin{tikzpicture}
 \drawTensorNode{(0,0)}{}{}
 \drawCorneredOpen{<-}{TensorIL}{(-1,2)}
 \drawCorneredOpen{<-}{TensorIR}{(1,2)}
 \drawStraightOpen{->}{TensorO}{(0,-2)}
\end{tikzpicture}
\fi}
\newcommand{\figGeneratorsParr}{\ifshowFigGenerators
\begin{tikzpicture}
 \drawParrNode{(0,0)}{}{}
 \drawCorneredOpen{<-}{ParrIL}{(-1,2)}
 \drawCorneredOpen{<-}{ParrIR}{(1,2)}
 \drawStraightOpen{->}{ParrO}{(0,-2)}
\end{tikzpicture}
\fi}
\newcommand{\figGeneratorsPrincipal}{\ifshowFigGenerators
\begin{tikzpicture}
 \drawPrincipalNode{(0,0)}{}{}
 \drawStraightOpen{<-}{PrincipalI}{(0,2)}
 \drawStraightOpen{->}{PrincipalO}{(0,-2)}
\end{tikzpicture}
\fi}
\newcommand{\figGeneratorsAuxiliary}{\ifshowFigGenerators
\begin{tikzpicture}
 \drawAuxiliaryNode{(0,0)}{}{}
 \drawStraightOpen{<-}{AuxiliaryI}{(0,2)}
 \drawStraightOpen{->}{AuxiliaryO}{(0,-2)}
\end{tikzpicture}
\fi}
\newcommand{\figGeneratorsDereliction}{\ifshowFigGenerators
\begin{tikzpicture}
 \drawDerelictionNode{(0,0)}{}{}
 \drawStraightOpen{<-}{DerelictionI}{(0,2)}
 \drawStraightOpen{->}{DerelictionO}{(0,-2)}
\end{tikzpicture}
\fi}

\newcommand{\figGeneratorsContraction}{\ifshowFigGenerators
\begin{tikzpicture}
 \drawContractionNode{n}{(0,0)}{1.5}{1}{}{}
 \drawStraightOpen{<-}{ContractionI}{(0,2)} \strikedOut
 \drawStraightOpen{->}{ContractionO}{(0,-2)}
\end{tikzpicture}
\fi}

\newcommand{\figBox}{\ifshowFigBox
\begin{tikzpicture}
 \draw (-3,3) rectangle (3,6) node [midway] {$H$};
 \draw [dashed] (-4,0) rectangle (4,7);
 \drawAuxiliaryNode{(-2,0)}{}{} \thickened
 \drawPrincipalNode{(2,0)}{}{}
 \drawStraightOpen{<-}{AuxiliaryI}{(0,2)} \strikedOut
 \drawStraightOpen{->}{AuxiliaryO}{(0,-2)} \strikedOut
 \drawStraightOpen{<-}{PrincipalI}{(0,2)}
 \drawStraightOpen{->}{PrincipalO}{(0,-2)}
\end{tikzpicture}
\fi}

\newcommand{\figPassTransitionsAxL}{\ifshowFigPassTransitions
\begin{tikzpicture}
 \drawAxNode{(0,0)}{$\alpha$}{}
 \drawCorneredOpen{}{AxL}{(-1.5,-2)} \withUpToken{left}
 \drawCorneredOpen{}{AxR}{(1.5,-2)}
\end{tikzpicture}
\fi}
\newcommand{\figPassTransitionsAxR}{\ifshowFigPassTransitions
\begin{tikzpicture}
 \drawAxNode{(0,0)}{$\alpha$}{}
 \drawCorneredOpen{}{AxL}{(-1.5,-2)}
 \drawCorneredOpen{}{AxR}{(1.5,-2)} \withDownToken{right}
\end{tikzpicture}
\fi}
\newcommand{\figPassTransitionsCutL}{\ifshowFigPassTransitions
\begin{tikzpicture}
 \drawCutNode{(0,0)}{$\alpha$}{}
 \drawCorneredOpen{}{CutL}{(-1,2)} \withDownToken{left}
 \drawCorneredOpen{}{CutR}{(1,2)}
\end{tikzpicture}
\fi}
\newcommand{\figPassTransitionsCutR}{\ifshowFigPassTransitions
\begin{tikzpicture}
 \drawCutNode{(0,0)}{$\alpha$}{}
 \drawCorneredOpen{}{CutL}{(-1,2)}
 \drawCorneredOpen{}{CutR}{(1,2)} \withUpToken{right}
\end{tikzpicture}
\fi}
\newcommand{\figPassTransitionsMultLeftDownL}{\ifshowFigPassTransitions
\begin{tikzpicture}
 \drawDollarNode{(0,0)}{$\alpha$}{}
 \drawCorneredOpen{}{DollarIL}{(-1,2)} \withDownToken{left}
 \drawCorneredOpen{}{DollarIR}{(1,2)}
 \drawStraightOpen{}{DollarO}{(0,-2)}
\end{tikzpicture}
\fi}
\newcommand{\figPassTransitionsMultLeftDownR}{\ifshowFigPassTransitions
\begin{tikzpicture}
 \drawDollarNode{(0,0)}{$\alpha$}{}
 \drawCorneredOpen{}{DollarIL}{(-1,2)}
 \drawCorneredOpen{}{DollarIR}{(1,2)}
 \drawStraightOpen{}{DollarO}{(0,-2)} \withDownToken{left}
\end{tikzpicture}
\fi}
\newcommand{\figPassTransitionsMultRightDownL}{\ifshowFigPassTransitions
\begin{tikzpicture}
 \drawDollarNode{(0,0)}{$\alpha$}{}
 \drawCorneredOpen{}{DollarIL}{(-1,2)}
 \drawCorneredOpen{}{DollarIR}{(1,2)} \withDownToken{right}
 \drawStraightOpen{}{DollarO}{(0,-2)}
\end{tikzpicture}
\fi}
\newcommand{\figPassTransitionsMultRightDownR}{\ifshowFigPassTransitions
\begin{tikzpicture}
 \drawDollarNode{(0,0)}{$\alpha$}{}
 \drawCorneredOpen{}{DollarIL}{(-1,2)}
 \drawCorneredOpen{}{DollarIR}{(1,2)}
 \drawStraightOpen{}{DollarO}{(0,-2)} \withDownToken{right}
\end{tikzpicture}
\fi}
\newcommand{\figPassTransitionsMultLeftUpL}{\ifshowFigPassTransitions
\begin{tikzpicture}
 \drawDollarNode{(0,0)}{$\alpha$}{}
 \drawCorneredOpen{}{DollarIL}{(-1,2)}
 \drawCorneredOpen{}{DollarIR}{(1,2)}
 \drawStraightOpen{}{DollarO}{(0,-2)} \withUpToken{left}
\end{tikzpicture}
\fi}
\newcommand{\figPassTransitionsMultLeftUpR}{\ifshowFigPassTransitions
\begin{tikzpicture}
 \drawDollarNode{(0,0)}{$\alpha$}{}
 \drawCorneredOpen{}{DollarIL}{(-1,2)} \withUpToken{left}
 \drawCorneredOpen{}{DollarIR}{(1,2)}
 \drawStraightOpen{}{DollarO}{(0,-2)}
\end{tikzpicture}
\fi}
\newcommand{\figPassTransitionsMultRightUpL}{\ifshowFigPassTransitions
\begin{tikzpicture}
 \drawDollarNode{(0,0)}{$\alpha$}{}
 \drawCorneredOpen{}{DollarIL}{(-1,2)}
 \drawCorneredOpen{}{DollarIR}{(1,2)}
 \drawStraightOpen{}{DollarO}{(0,-2)} \withUpToken{right}
\end{tikzpicture}
\fi}
\newcommand{\figPassTransitionsMultRightUpR}{\ifshowFigPassTransitions
\begin{tikzpicture}
 \drawDollarNode{(0,0)}{$\alpha$}{}
 \drawCorneredOpen{}{DollarIL}{(-1,2)}
 \drawCorneredOpen{}{DollarIR}{(1,2)} \withUpToken{right}
 \drawStraightOpen{}{DollarO}{(0,-2)}
\end{tikzpicture}
\fi}
\newcommand{\figPassTransitionsPrincipalL}{\ifshowFigPassTransitions
\begin{tikzpicture}
 \drawPrincipalNode{(0,0)}{$\alpha$}{}
 \drawStraightOpen{}{PrincipalI}{(0,2)} \withDownToken{left}
 \drawStraightOpen{}{PrincipalO}{(0,-2)}
\end{tikzpicture}
\fi}
\newcommand{\figPassTransitionsPrincipalR}{\ifshowFigPassTransitions
\begin{tikzpicture}
 \drawPrincipalNode{(0,0)}{$\alpha$}{}
 \drawStraightOpen{}{PrincipalI}{(0,2)}
 \drawStraightOpen{}{PrincipalO}{(0,-2)} \withDownToken{left}
\end{tikzpicture}
\fi}
\newcommand{\figPassTransitionsDerelictionL}{\ifshowFigPassTransitions
\begin{tikzpicture}
 \drawDerelictionNode{(0,0)}{$\alpha$}{}
 \drawStraightOpen{}{DerelictionI}{(0,2)} \withDownToken{left}
 \drawStraightOpen{}{DerelictionO}{(0,-2)}
\end{tikzpicture}
\fi}
\newcommand{\figPassTransitionsDerelictionR}{\ifshowFigPassTransitions
\begin{tikzpicture}
 \drawDerelictionNode{(0,0)}{$\alpha$}{}
 \drawStraightOpen{}{DerelictionI}{(0,2)}
 \drawStraightOpen{}{DerelictionO}{(0,-2)} \withDownToken{left}
\end{tikzpicture}
\fi}
\newcommand{\figPassTransitionsContractionL}{\ifshowFigPassTransitions
\begin{tikzpicture}
 \drawContractionNode{n}{(0,0)}{1.5}{1}{$\alpha$}{}
 \transAnchor{ContractionI}{(0,2)}{Midpoint}
 \drawStraightConnection{}{ContractionI}{Midpoint} \withDownToken{left}
 \drawStraightOpen{}{Midpoint}{(0,1)} \strikedOut
 \drawStraightOpen{}{ContractionO}{(0,-2)}
\end{tikzpicture}
\fi}
\newcommand{\figPassTransitionsContractionR}{\ifshowFigPassTransitions
\begin{tikzpicture}
 \drawContractionNode{n}{(0,0)}{1.5}{1}{$\alpha$}{}
 \drawStraightOpen{}{ContractionI}{(0,2)} \strikedOut
 \drawStraightOpen{}{ContractionO}{(0,-2)} \withDownToken{left}
\end{tikzpicture}
\fi}

\newcommand{\figRewriteTransitionsAxL}{\ifshowFigRewriteTransitions
\begin{tikzpicture}
 \drawCutNode{(-2.5,-1)}{$\alpha$}{}
 \drawAxNode{(2.5,1)}{$\beta$}{}
 \midAnchor{Cut}{Ax}{Midpoint}
 \drawCorneredOpen{}{CutL}{(-1.5,2)} \withUpToken{left}
 \drawCorneredConnection{}{CutR}{Midpoint}
 \drawCorneredConnection{}{AxL}{Midpoint}
 \drawCorneredOpen{}{AxR}{(1.5,-2)}
\end{tikzpicture}
\fi}
\newcommand{\figRewriteTransitionsAxR}{\ifshowFigRewriteTransitions
\begin{tikzpicture}
 \putAnchor{(-2,0)}{Bottom}
 \drawStraightOpen{}{Bottom}{(0,4)} \withUpToken{right}
\end{tikzpicture}
\fi}
\newcommand{\figRewriteTransitionsAxPrincipalL}{\ifshowFigRewriteTransitions
\begin{tikzpicture}
 \drawPrincipalNode{(-6,1)}{$\alpha$}{}
 \drawCutNode{(-2.5,-1)}{$\beta$}{}
 \drawAxNode{(2.5,1)}{$\gamma$}{}
 \midAnchor{Cut}{Ax}{Midpoint}
 \drawStraightOpen{}{PrincipalI}{(0,2)} \withUpToken{left}
 \drawCorneredConnection{}{CutL}{PrincipalO}
 \drawCorneredConnection{}{CutR}{Midpoint}
 \drawCorneredConnection{}{AxL}{Midpoint}
 \drawCorneredOpen{}{AxR}{(1.5,-2)}
\end{tikzpicture}
\fi}
\newcommand{\figRewriteTransitionsAxPrincipalR}{\ifshowFigRewriteTransitions
\begin{tikzpicture}
 \drawPrincipalNode{(0,0)}{$\alpha$}{}
 \drawStraightOpen{}{PrincipalI}{(0,2)} \withUpToken{left}
 \drawStraightOpen{}{PrincipalO}{(0,-2)}
\end{tikzpicture}
\fi}
\newcommand{\figRewriteTransitionsMultLeftL}{\ifshowFigRewriteTransitions
\begin{tikzpicture}
 \drawParrNode{(-3,1)}{$\alpha$}{}
 \drawCutNode{(0,-1)}{$\beta$}{}
 \drawTensorNode{(3,1)}{$\gamma$}{}
 \drawCorneredOpen{}{ParrIL}{(-1,2)} \withUpToken{left}
 \drawCorneredOpen{}{ParrIR}{(1,2)}
 \drawCorneredConnection{}{CutL}{ParrO}
 \drawCorneredConnection{}{CutR}{TensorO}
 \drawCorneredOpen{}{TensorIL}{(-1,2)}
 \drawCorneredOpen{}{TensorIR}{(1,2)}
\end{tikzpicture}
\fi}
\newcommand{\figRewriteTransitionsMultLeftR}{\ifshowFigRewriteTransitions
\begin{tikzpicture}
 \drawCutNode{(0,0)}{$\beta$}{On}
 \drawCutNode{(0,-3)}{$\nu$}{Off}
 \drawCorneredOpen{}{CutOnL}{(-3.5,2)} \withUpToken{right}
 \drawCorneredOpen{}{CutOnR}{(1,2)}
 \drawCorneredOpen{}{CutOffL}{(-1,5)}
 \drawCorneredOpen{}{CutOffR}{(3,5)}
\end{tikzpicture}
\fi}
\newcommand{\figRewriteTransitionsMultRightL}{\ifshowFigRewriteTransitions
\begin{tikzpicture}
 \drawParrNode{(-3,1)}{$\alpha$}{}
 \drawCutNode{(0,-1)}{$\beta$}{}
 \drawTensorNode{(3,1)}{$\gamma$}{}
 \drawCorneredOpen{}{ParrIL}{(-1,2)}
 \drawCorneredOpen{}{ParrIR}{(1,2)} \withUpToken{left}
 \drawCorneredConnection{}{CutL}{ParrO}
 \drawCorneredConnection{}{CutR}{TensorO}
 \drawCorneredOpen{}{TensorIL}{(-1,2)}
 \drawCorneredOpen{}{TensorIR}{(1,2)}
\end{tikzpicture}
\fi}
\newcommand{\figRewriteTransitionsMultRightR}{\ifshowFigRewriteTransitions
\begin{tikzpicture}
 \drawCutNode{(0,0)}{$\beta$}{On}
 \drawCutNode{(0,-3)}{$\nu$}{Off}
 \drawCorneredOpen{}{CutOnL}{(-1,2)} \withUpToken{left}
 \drawCorneredOpen{}{CutOnR}{(3,2)}
 \drawCorneredOpen{}{CutOffL}{(-3.5,5)}
 \drawCorneredOpen{}{CutOffR}{(1,5)}
\end{tikzpicture}
\fi}
\newcommand{\figRewriteTransitionsDerelictionL}{\ifshowFigRewriteTransitions
\begin{tikzpicture}
 \draw (-3,3) rectangle (3,6) node [midway] {$H$};
 \draw [dashed] (-4,0) rectangle (4,7);
 \drawAuxiliaryNode{(-2,0)}{$\vec{\delta}$}{} \thickened
 \drawPrincipalNode{(2,0)}{$\alpha$}{}
 \drawCutNode{(5,-2)}{$\beta$}{}
 \drawDerelictionNode{(8,0)}{$\gamma$}{}
 \drawStraightOpen{}{AuxiliaryI}{(0,2)} \strikedOut
 \drawStraightOpen{}{AuxiliaryO}{(0,-2)} \strikedOut
 \drawStraightOpen{}{PrincipalI}{(0,2)} \withUpToken{left}
 \drawCorneredConnection{}{CutL}{PrincipalO}
 \drawCorneredConnection{}{CutR}{DerelictionO}
 \drawStraightOpen{}{DerelictionI}{(0,2)}
\end{tikzpicture}
\fi}
\newcommand{\figRewriteTransitionsDerelictionR}{\ifshowFigRewriteTransitions
\begin{tikzpicture}
 \draw (-2,3) rectangle (2,6) node [midway] {$H$};
 \putAnchor{(-1.5,3)}{Auxiliary}
 \putAnchor{(1.5,3)}{Principal}
 \drawCutNode{(4,1)}{$\beta$}{}
 \drawStraightOpen{}{Auxiliary}{(0,-3)} \strikedOut
 \drawCorneredConnection{}{CutL}{Principal} \withUpToken{left}
 \drawCorneredOpen{}{CutR}{(1,3)}
\end{tikzpicture}
\fi}
\newcommand{\figRewriteTransitionsContractionPosL}{\ifshowFigRewriteTransitions
\begin{tikzpicture}
 \draw (-3,3) rectangle (3,6) node [midway] {$H$};
 \draw [dashed] (-4,0) rectangle (4,7);
 \drawAuxiliaryNode{(-2,0)}{$\vec{\epsilon}$}{} \thickened
 \drawPrincipalNode{(2,0)}{$\alpha$}{}
 \drawContractionNode{g+f^{-1}}{(-2,-4)}{3.6}{1.3}{$\vec{\phi}$}{Aux} \thickened
 \drawCutNode{(5,-2)}{$\beta$}{}
 \drawContractionNode{n+1}{(8,1)}{2.5}{1}{$\gamma$}{}
 \drawDummyNode{(10,4)}{$\delta$}{}
 \drawStraightOpen{}{AuxiliaryI}{(0,2)} \strikedOut
 \drawStraightOpen{}{PrincipalI}{(0,2)} \withUpToken{left}
 \drawCorneredConnection{}{CutL}{PrincipalO}
 \drawCorneredConnection{}{CutR}{ContractionO}
 \drawCorneredOpen{}{ContractionIL}{(-0.5,2)} \strikedOut
 \drawCorneredConnection{}{ContractionIR}{DummyO}
 \drawCorneredOpen{}{ContractionAuxIL}{(-1,2)} \strikedOut
 \drawStraightConnection{}{ContractionAuxI}{AuxiliaryO} \strikedOut
 \drawStraightOpen{}{ContractionAuxO}{(0,-2)} \strikedOut
\end{tikzpicture}
\fi}
\newcommand{\figRewriteTransitionsContractionPosR}{\ifshowFigRewriteTransitions
\begin{tikzpicture}
 \draw (-10,14) rectangle (-5,17) node [midway] {$H^\approx$};
 \draw [dashed] (-11,11) rectangle (-4,18);
 \drawAuxiliaryNode{(-9,11)}{$\vec{\mu}$}{Copy} \thickened
 \drawPrincipalNode{(-6,11)}{$\nu$}{Copy}
 \draw (-7,3) rectangle (-2,6) node [midway] {$H$};
 \draw [dashed] (-8,0) rectangle (-1,7);
 \drawAuxiliaryNode{(-6,0)}{$\vec{\epsilon}$}{} \thickened
 \drawPrincipalNode{(-3,0)}{$\alpha$}{}
 \drawContractionNode{g+2f^{-1}}{(-9,-4)}{4.1}{1.3}{$\vec{\phi}$}{Aux} \thickened
 \drawCutNode{(0,-2)}{$\beta$}{Off}
 \drawCutNode{(-3,9)}{$\eta$}{On}
 \drawContractionNode{n}{(3,0)}{1.5}{1}{$\gamma$}{}
 \drawDummyNode{(0,11)}{$\delta$}{}
 \drawStraightOpen{}{AuxiliaryCopyI}{(0,2)} \strikedOut
 \drawStraightOpen{}{PrincipalCopyI}{(0,2)} \withUpToken{left}
 \drawCorneredConnection{}{CutOnL}{PrincipalCopyO}
 \drawCorneredConnection{}{CutOnR}{DummyO}
 \drawStraightOpen{}{AuxiliaryI}{(0,2)} \strikedOut
 \drawStraightOpen{}{PrincipalI}{(0,2)}
 \drawCorneredConnection{}{CutOffL}{PrincipalO}
 \drawCorneredConnection{}{CutOffR}{ContractionO}
 \drawStraightOpen{}{ContractionI}{(0,2)} \strikedOut
 \drawCorneredOpen{}{ContractionAuxIL}{(-1,2)} \strikedOut
 \drawStraightConnection{}{ContractionAuxI}{AuxiliaryCopyO} \strikedOut
 \drawCorneredConnection{}{ContractionAuxIR}{AuxiliaryO} \strikedOut
 \drawStraightOpen{}{ContractionAuxO}{(0,-2)} \strikedOut
\end{tikzpicture}
\fi}
\newcommand{\figRewriteTransitionsContractionNullL}{\ifshowFigRewriteTransitions
\begin{tikzpicture}
 \draw (-3,3) rectangle (3,6) node [midway] {$H$};
 \draw [dashed] (-4,0) rectangle (4,7);
 \drawAuxiliaryNode{(-2,0)}{$\vec{\delta}$}{} \thickened
 \drawPrincipalNode{(2,0)}{$\alpha$}{}
 \drawCutNode{(5,-2)}{$\beta$}{}
 \drawContractionNode{1}{(8,0)}{1.2}{1.2}{$\gamma$}{}
 \drawDummyNode{(8,4)}{$\delta$}{}
 \drawStraightOpen{}{AuxiliaryI}{(0,2)} \strikedOut
 \drawStraightOpen{}{AuxiliaryO}{(0,-2)} \strikedOut
 \drawStraightOpen{}{PrincipalI}{(0,2)} \withUpToken{left}
 \drawCorneredConnection{}{CutL}{PrincipalO}
 \drawCorneredConnection{}{CutR}{ContractionO}
 \drawStraightConnection{}{ContractionI}{DummyO}
\end{tikzpicture}
\fi}
\newcommand{\figRewriteTransitionsContractionNullR}{\ifshowFigRewriteTransitions
\begin{tikzpicture}
 \draw (-3,3) rectangle (3,6) node [midway] {$H$};
 \draw [dashed] (-4,0) rectangle (4,7);
 \drawAuxiliaryNode{(-2,0)}{$\vec{\delta}$}{} \thickened
 \drawPrincipalNode{(2,0)}{$\alpha$}{}
 \drawCutNode{(5,-2)}{$\beta$}{}
 \drawDummyNode{(8,4)}{$\delta$}{}
 \drawStraightOpen{}{AuxiliaryI}{(0,2)} \strikedOut
 \drawStraightOpen{}{AuxiliaryO}{(0,-2)} \strikedOut
 \drawStraightOpen{}{PrincipalI}{(0,2)} \withUpToken{left}
 \drawCorneredConnection{}{CutL}{PrincipalO}
 \drawCorneredConnection{}{CutR}{DummyO}
\end{tikzpicture}
\fi}

\newcommand{\defTranslationsTerm}{\ifshowDefTranslationsTerm
\begin{tikzpicture}
 \draw (-2,0) rectangle (2,3) node [midway] {$t^\dag$};
 \putAnchor{(-1,0)}{TFreeO}
 \putAnchor{(1,0)}{TPriO}
 \drawStraightOpen{}{TFreeO}{(0,-2)} \strikedOut \withLabel{east}{\far}{$\FV(t)$}
 \drawStraightOpen{}{TPriO}{(0,-2)}
\end{tikzpicture}
\fi}
\newcommand{\defTranslationsEvalCtxt}{\ifshowDefTranslationsTerm
\begin{tikzpicture}
 \draw (-2,0) rectangle (2,3) node [midway] {$E_M^\dag$};
 \putAnchor{(-1,3)}{EFreeI}
 \putAnchor{(1,3)}{EPriI}
 \putAnchor{(-1,0)}{EFreeO}
 \putAnchor{(1,0)}{EPriO}
 \drawStraightOpen{}{EFreeI}{(0,2)} \strikedOut \withLabel{east}{\far}{$M$}
 \drawStraightOpen{}{EPriI}{(0,2)}
 \drawStraightOpen{}{EFreeO}{(0,-2)} \strikedOut \withLabel{east}{\far}{$\FV_M(E)$}
 \drawStraightOpen{}{EPriO}{(0,-2)}
\end{tikzpicture}
\fi}

\newcommand{\lemTranslationCompositionTermConnected}{\ifshowLemTranslationComposition
\begin{tikzpicture}
 \draw (-4,5) rectangle (4,8) node [midway] {$\overline{t}^\dag$};
 \putAnchor{(-3,5)}{TFreeO}
 \putAnchor{(3,5)}{TPriO}
 \draw (-4,0) rectangle (4,3) node [midway] {$E_{\FV(\overline{t})}^\dag$};
 \putAnchor{(-3,3)}{EFreeI}
 \putAnchor{(3,3)}{EPriI}
 \putAnchor{(-3,0)}{EFreeO}
 \putAnchor{(3,0)}{EPriO}
 \drawStraightConnection{}{TFreeO}{EFreeI} \strikedOut \withLabel{east}{\far}{$\FV(t)$}
 \drawStraightConnection{}{TPriO}{EPriI}
 \drawStraightOpen{}{EFreeO}{(0,-2)} \strikedOut \withLabel{east}{\far}{$\FV_{\FV(t)}(E)$}
 \drawStraightOpen{}{EPriO}{(0,-2)}
\end{tikzpicture}
\fi}

\newcommand{\lemTranslationCompositionEvalCtxtConnected}{\ifshowLemTranslationComposition
\begin{tikzpicture}
 \draw (-4,5) rectangle (4,8) node [midway] {$(E')_M^\dag$};
 \putAnchor{(-3,8)}{EUpperFreeI}
 \putAnchor{(3,8)}{EUpperPriI}
 \putAnchor{(-3,5)}{EUpperFreeO}
 \putAnchor{(3,5)}{EUpperPriO}
 \draw (-5,0) rectangle (5,3) node [midway] {$E_{\FV_M(E')}^\dag$};
 \putAnchor{(-3,3)}{ELowerFreeI}
 \putAnchor{(3,3)}{ELowerPriI}
 \putAnchor{(-3,0)}{ELowerFreeO}
 \putAnchor{(3,0)}{ELowerPriO}
 \drawStraightOpen{}{EUpperFreeI}{(0,2)} \strikedOut \withLabel{east}{\far}{$M$}
 \drawStraightOpen{}{EUpperPriI}{(0,2)}
 \drawStraightConnection{}{EUpperFreeO}{ELowerFreeI} \strikedOut \withLabel{east}{\far}{$\FV_M(E)$}
 \drawStraightConnection{}{EUpperPriO}{ELowerPriI}
 \drawStraightOpen{}{ELowerFreeO}{(0,-2)} \strikedOut \withLabel{east}{\far}{$\FV_{\FV_M(E')}(E)$}
 \drawStraightOpen{}{ELowerPriO}{(0,-2)}
\end{tikzpicture}
\fi}

\newcommand{\figTranslationTermsVar}{\ifshowFigTranslationTerms
\begin{tikzpicture}
 \drawAxNode{(0,0)}{}{}
 \drawCorneredOpen{}{AxL}{(-1,-2)} \withLabel{north east}{\near}{$x$}
 \drawCorneredOpen{}{AxR}{(1,-2)}
\end{tikzpicture}
\fi}
\newcommand{\figTranslationTermsAbsCont}{\ifshowFigTranslationTerms
\begin{tikzpicture}
 \draw (-4,9) rectangle (4,12) node [midway] {$t^\dag$};
 \putAnchor{(-3,9)}{TFreeO}
 \putAnchor{(0,9)}{TBoundO}
 \putAnchor{(3,9)}{TPriO}
 \draw [dashed] (-12,0) rectangle (5,13);
 \drawAuxiliaryNode{(-3,0)}{}{} \thickened
 \drawPrincipalNode{(1.5,0)}{}{}
 \drawContractionNode{k}{(0,6)}{1.5}{1}{}{}
 \drawTensorNode{(1.5,3)}{}{}
 \drawStraightConnection{}{TFreeO}{AuxiliaryI} \strikedOut \withLabel{east}{\far}{$\FV(t) \backslash x$}
 \drawStraightOpen{}{AuxiliaryO}{(0,-2)} \strikedOut \withLabel{east}{\far}{$\FV(t) \backslash x$}
 \drawStraightConnection{}{TBoundO}{ContractionI} \strikedOut \withLabel{east}{\far}{$\mathbf{x}$}
 \drawCorneredConnection{}{TensorIL}{ContractionO} \withLabel{east}{\near}{$x$}
 \drawCorneredConnection{}{TensorIR}{TPriO}
 \drawStraightConnection{}{TensorO}{PrincipalI}
 \drawStraightOpen{}{PrincipalO}{(0,-2)}
\end{tikzpicture}
\fi}

\newcommand{\figTranslationTermsApp}{\ifshowFigTranslationTerms
\begin{tikzpicture}
 \draw (-7,6) rectangle (-3,9) node [midway] {$t^\dag$};
 \putAnchor{(-6,6)}{TFreeO}
 \putAnchor{(-5,6)}{TPriO}
 \draw (0,6) rectangle (4,9) node [midway] {$u^\dag$};
 \putAnchor{(2,6)}{UFreeO}
 \putAnchor{(3,6)}{UPriO}
 \drawCutNode{(-2,0)}{}{}
 \drawTensorNode{(5,5)}{}{}
 \drawDerelictionNode{(5,2)}{}{}
 \drawAxNode{(9,7)}{}{}
 \midAnchor{TensorIR}{AxL}{Midpoint}
 \drawStraightOpen{}{TFreeO}{(0,-8)} \strikedOut \withLabel{east}{\far}{$\FV(t)$}
 \drawStraightOpen{}{UFreeO}{(0,-8)} \strikedOut \withLabel{south east}{\near}{$\FV(u)$}
 \drawCorneredConnection{}{CutL}{TPriO}
 \drawCorneredConnection{}{CutR}{DerelictionO}
 \drawStraightConnection{}{TensorO}{DerelictionI}
 \drawCorneredConnection{}{TensorIL}{UPriO}
 \drawCorneredConnection{}{TensorIR}{Midpoint}
 \drawCorneredConnection{}{AxL}{Midpoint}
 \drawCorneredOpen{}{AxR}{(1,-10)}
\end{tikzpicture}
\fi}
\newcommand{\figTranslationTermsESCont}{\ifshowFigTranslationTerms
\begin{tikzpicture}
 \draw (-10,5) rectangle (-3,8) node [midway] {$t^\dag$};
 \putAnchor{(-9,5)}{TFreeO}
 \putAnchor{(-6,5)}{TBoundO}
 \putAnchor{(-4,5)}{TPriO}
 \draw (0,5) rectangle (5,8) node [midway] {$u^\dag$};
 \putAnchor{(2,5)}{UFreeO}
 \putAnchor{(4,5)}{UPriO}
 \drawContractionNode{k}{(-6,2)}{1.5}{1}{}{}
 \drawCutNode{(-1,0)}{}{}
 \putAnchor{(0,-1.5)}{Midpoint}
 \drawStraightOpen{}{TFreeO}{(0,-8)} \strikedOut \withLabel{east}{\far}{$\FV(t) \backslash x$}
 \drawStraightConnection{}{TBoundO}{ContractionI} \strikedOut \withLabel{east}{\far}{$\mathbf{x}$}
 \drawStraightOpen{}{TPriO}{(0,-8)}
 \drawCorneredConnection{}{CutL}{ContractionO} \withLabel{east}{\near}{$x$}
 \drawCorneredConnection{}{CutR}{UPriO}
 \drawCorneredConnection{}{Midpoint}{UFreeO} \strikedOut \withLabel{north}{\far}{$\FV(u)$}
 \drawCorneredOpen{}{Midpoint}{(-6,-1.5)}
\end{tikzpicture}
\fi}

\newcommand{\defSimulationTermConfig}{\ifshowDefSimulation
\begin{tikzpicture}
 \draw (-4,5) rectangle (4,8) node [midway] {$\overline{t}^\dag$};
 \putAnchor{(-3,5)}{TFreeO}
 \putAnchor{(3,5)}{TPriO}
 \draw (-4,0) rectangle (4,3) node [midway] {$E_{\FV(\overline{t})}^\dag$};
 \putAnchor{(-3,3)}{EFreeI}
 \putAnchor{(3,3)}{EPriI}
 \putAnchor{(3,0)}{EPriO}
 \drawStraightConnection{}{TFreeO}{EFreeI} \strikedOut \withLabel{east}{\far}{$\FV(\overline{t})$}
 \drawStraightConnection{}{TPriO}{EPriI} \withUpToken{left}
 \drawStraightOpen{}{EPriO}{(0,-2)}
\end{tikzpicture}
\fi}
\newcommand{\defSimulationCtxtConfig}{\ifshowDefSimulation
\begin{tikzpicture}
 \draw (-4,9) rectangle (4,12) node [midway] {$\overline{v}^\diamond$};
 \putAnchor{(-3,9)}{TFreeO}
 \putAnchor{(3,9)}{TPriO}
 \draw [dashed] (-10,6) rectangle (5,13);
 \draw (-4,0) rectangle (4,3) node [midway] {$E_{\FV(\overline{v})}^\dag$};
 \putAnchor{(-3,3)}{EFreeI}
 \putAnchor{(3,3)}{EPriI}
 \putAnchor{(3,0)}{EPriO}
 \drawAuxiliaryNode{(-3,6)}{}{} \thickened
 \drawPrincipalNode{(3,6)}{}{}
 \drawStraightConnection{}{TFreeO}{AuxiliaryI} \strikedOut \withLabel{east}{\far}{$\FV(\overline{v})$}
 \drawStraightConnection{}{TPriO}{PrincipalI} \withUpToken{left}
 \drawStraightConnection{}{AuxiliaryO}{EFreeI} \strikedOut \withLabel{east}{\far}{$\FV(\overline{v})$}
 \drawStraightConnection{}{PrincipalO}{EPriI}
 \drawStraightOpen{}{EPriO}{(0,-2)}
\end{tikzpicture}
\fi}
\newcommand{\defSimulationValue}{\ifshowDefSimulation
\begin{tikzpicture}
 \draw (-2.5,8) rectangle (2.5,11) node [midway] {$\overline{v}^\diamond$};
 \putAnchor{(-1.5,8)}{TFreeO}
 \putAnchor{(1.5,8)}{TPriO}
 \draw [dashed] (-8.5,5) rectangle (3.5,12);
 \drawAuxiliaryNode{(-1.5,5)}{}{} \thickened
 \drawPrincipalNode{(1.5,5)}{}{}
 \drawStraightConnection{}{TFreeO}{AuxiliaryI} \strikedOut \withLabel{east}{\far}{$\FV(\overline{v})$}
 \drawStraightConnection{}{TPriO}{PrincipalI}
 \drawStraightOpen{}{AuxiliaryO}{(0,-2)} \strikedOut \withLabel{east}{\far}{$\FV(\overline{v})$}
 \drawStraightOpen{}{PrincipalO}{(0,-2)}
\end{tikzpicture}
\fi}

\newcommand{\lemDecomposeTranslationSubstCtxt}{\ifshowLemTranslationEvalCtxt
\begin{tikzpicture}
 \draw (-2,0) rectangle (2,3) node [midway] {$A_M^\ddag$};
 \putAnchor{(0,3)}{AFreeI}
 \putAnchor{(0,0)}{AFreeO}
 \putAnchor{(3,5)}{WiresI}
 \drawStraightOpen{}{WiresI}{(0,-7)}
 \drawStraightOpen{}{AFreeI}{(0,2)} \strikedOut \withLabel{east}{\far}{$M$}
 \drawStraightOpen{}{AFreeO}{(0,-2)} \strikedOut \withLabel{east}{\far}{$\FV_M(A)$}
\end{tikzpicture}
\fi}
\newcommand{\lemDecomposeTranslationEvalCtxtNotBound}{\ifshowLemTranslationEvalCtxt
\begin{tikzpicture}
 \draw (-2,0) rectangle (2,3) node [midway] {$E_M^\dag$};
 \putAnchor{(-1,3)}{EFreeI}
 \putAnchor{(1,3)}{EPriI}
 \putAnchor{(-1,0)}{EFreeO}
 \putAnchor{(1,0)}{EPriO}
 \putAnchor{(-11,5)}{WiresI}
 \drawStraightOpen{}{WiresI}{(0,-7)} \strikedOut \withLabel{east}{\far}{$M_0$}
 \drawStraightOpen{}{EFreeI}{(0,2)} \strikedOut \withLabel{east}{\far}{$M$}
 \drawStraightOpen{}{EPriI}{(0,2)}
 \drawStraightOpen{}{EFreeO}{(0,-2)} \strikedOut \withLabel{east}{\far}{$\FV_M(E)$}
 \drawStraightOpen{}{EPriO}{(0,-2)}
\end{tikzpicture}
\fi}
\newcommand{\lemDecomposeTranslationEvalCtxtBound}{\ifshowLemTranslationEvalCtxt
\begin{tikzpicture}
 \draw (-9,0) rectangle (10,3) node [midway] {$E_M^{\dag\dag}$};
 \putAnchor{(-8,3)}{EFreeI}
 \putAnchor{(9,3)}{EPriI}
 \putAnchor{(-8,0)}{EFreeO}
 \putAnchor{(2,0)}{EM1O}
 \putAnchor{(9,0)}{EPriO}
 \putAnchor{(-10,5)}{WiresI}
 \drawContractionNode{M_0 + M_1}{(2,-3)}{4.2}{1.2}{}{} \thickened
 \drawCutNode{(5,-6)}{}{} \thickened
 \putAnchor{(-9,-1)}{Midpoint}
 \drawStraightOpen{}{EFreeI}{(0,2)} \strikedOut \withLabel{west}{\far}{$M$}
 \drawStraightOpen{}{EPriI}{(0,2)}
 \drawStraightOpen{}{EFreeO}{(0,-8)} \strikedOut \withLabel{east}{\far}{$\FV_{M_0+M}(E)$}
 \drawStraightOpen{}{EPriO}{(0,-8)}
 \drawCorneredConnection{}{Midpoint}{WiresI} \strikedOut \withLabel{east}{\far}{$M_0$}
 \drawCorneredConnection{}{Midpoint}{ContractionIL}
 \drawStraightConnection{}{ContractionI}{EM1O} \strikedOut \withLabel{west}{\far}{$M_1$}
 \drawCorneredConnection{}{CutL}{ContractionO} \strikedOut \withLabel{east}{\near}{$\mathrm{supp}(M_0)$}
 \drawCorneredOpen{}{CutR}{(1,6)} \strikedOut
\end{tikzpicture}
\fi}

\newcommand{\lemDecomposeTranslationEvalCtxtProofFirst}{\ifshowPrfDecomposeTranslationEvalCtxt
\begin{tikzpicture}
 \draw (-12,5) rectangle (-1,8) node [midway] {$(E_2)_{M_0+M}^\dag$};
 \putAnchor{(-9,8)}{E2FreeM0I}
 \putAnchor{(-8,8)}{E2FreeMI}
 \putAnchor{(-4,8)}{E2PriI}
 \putAnchor{(-9,5)}{E2FreeO}
 \putAnchor{(-6,5)}{E2BoundO}
 \putAnchor{(-4,5)}{E2PriO}
 \draw (0,5) rectangle (5,8) node [midway] {$t^\dag$};
 \putAnchor{(2,5)}{TFreeO}
 \putAnchor{(4,5)}{TPriO}
 \drawContractionNode{k}{(-6,2)}{1.5}{1}{}{}
 \drawCutNode{(-1,0)}{}{}
 \putAnchor{(0,-1.5)}{Midpoint}
 \draw (-10,-6) rectangle (-3,-3) node [midway] {$(E_1)^\dag$};
 \putAnchor{(-9,-3)}{E1Free1I}
 \putAnchor{(-8,-3)}{E1Free2I}
 \putAnchor{(-4,-3)}{E1PriI}
 \putAnchor{(-9,-6)}{E1FreeO}
 \putAnchor{(-4,-6)}{E1PriO}
 \drawStraightOpen{}{E2FreeM0I}{(0,2)} \strikedOut \withLabel{east}{\far}{$M_0$}
 \drawStraightOpen{}{E2FreeMI}{(0,2)} \strikedOut \withLabel{west}{\far}{$M$}
 \drawStraightOpen{}{E2PriI}{(0,2)}
 \drawStraightConnection{}{E2FreeO}{E1Free1I} \strikedOut \withLabel{east}{\far}{$\FV_{M_0+M}(E_2) \backslash x$}
 \drawStraightConnection{}{E2BoundO}{ContractionI} \strikedOut \withLabel{east}{\far}{$\mathbf{x}$}
 \drawStraightConnection{}{E2PriO}{E1PriI}
 \drawCorneredConnection{}{CutL}{ContractionO} \withLabel{east}{\near}{$x$}
 \drawCorneredConnection{}{CutR}{TPriO}
 \drawCorneredConnection{}{Midpoint}{TFreeO} \strikedOut
 \drawCorneredConnection{}{Midpoint}{E1Free2I}
 \drawStraightOpen{}{E1FreeO}{(0,-2)} \strikedOut
 \drawStraightOpen{}{E1PriO}{(0,-2)}
\end{tikzpicture}
\fi}
\newcommand{\lemDecomposeTranslationEvalCtxtProofSecond}{\ifshowPrfDecomposeTranslationEvalCtxt
\begin{tikzpicture}
 \draw (-12,5) rectangle (-1,8) node [midway] {$(E_2)_{M_2+M}^\dag$};
 \putAnchor{(-9,8)}{E2FreeM2I}
 \putAnchor{(-8,8)}{E2FreeMI}
 \putAnchor{(-4,8)}{E2PriI}
 \putAnchor{(-9,5)}{E2FreeO}
 \putAnchor{(-6,5)}{E2BoundO}
 \putAnchor{(-4,5)}{E2PriO}
 \draw (1,5) rectangle (5,8) node [midway] {$t^\dag$};
 \putAnchor{(2,5)}{TFreeO}
 \putAnchor{(4,5)}{TPriO}
 \drawContractionNode{k}{(-6,2)}{1.5}{1}{}{}
 \drawCutNode{(-1,0)}{}{}
 \putAnchor{(0,-1.5)}{Midpoint}
 \putAnchor{(-16,10)}{M1I}
 \putAnchor{(-13.5,10)}{XsI}
 \putAnchor{(-12,4)}{XsMidpoint}
 \draw (-17,-6) rectangle (-3,-3) node [midway] {$(E_1)^\dag$};
 \putAnchor{(-16,-3)}{E1FreeM1I}
 \putAnchor{(-9,-3)}{E1FreeE2I}
 \putAnchor{(-8,-3)}{E1FreeTI}
 \putAnchor{(-4,-3)}{E1PriI}
 \putAnchor{(-16,-6)}{E1FreeO}
 \putAnchor{(-4,-6)}{E1PriO}
 \drawStraightOpen{}{E2FreeM2I}{(0,2)} \strikedOut \withLabel{east}{\far}{$M_2$}
 \drawStraightOpen{}{E2FreeMI}{(0,2)} \strikedOut \withLabel{west}{\far}{$M$}
 \drawStraightOpen{}{E2PriI}{(0,2)}
 \drawStraightConnection{}{E2FreeO}{E1FreeE2I} \strikedOut 
 \drawStraightConnection{}{E2BoundO}{ContractionI} \strikedOut \withLabel{west}{\near}{$\mathbf{x}$}
 \drawStraightConnection{}{E2PriO}{E1PriI}
 \drawCorneredConnection{}{XsMidpoint}{XsI} \strikedOut \withLabel{east}{\near}{$\mathbf{x}$}
 \drawCorneredConnection{}{XsMidpoint}{ContractionIL}
 \drawCorneredConnection{}{CutL}{ContractionO} \withLabel{east}{\near}{$x$}
 \drawCorneredConnection{}{CutR}{TPriO}
 \drawCorneredConnection{}{Midpoint}{TFreeO} \strikedOut
 \drawCorneredConnection{}{Midpoint}{E1FreeTI}
 \drawStraightConnection{}{M1I}{E1FreeM1I} \strikedOut \withLabel{east}{\far}{$M_1$}
 \drawStraightOpen{}{E1FreeO}{(0,-2)} \strikedOut
 \drawStraightOpen{}{E1PriO}{(0,-2)}
\end{tikzpicture}
\fi}
\newcommand{\lemDecomposeTranslationEvalCtxtProofThird}{\ifshowPrfDecomposeTranslationEvalCtxt
\begin{tikzpicture}
 \draw (-16,5) rectangle (-2,8) node [midway] {$(E_2)_M^{\dag\dag}$};
 \putAnchor{(-8,8)}{E2FreeMI}
 \putAnchor{(-4,8)}{E2PriI}
 \putAnchor{(-15,5)}{E2FreeO}
 \putAnchor{(-13,5)}{E2M2OL}
 \putAnchor{(-8,5)}{E2M2OR}
 \putAnchor{(-6,5)}{E2BoundO}
 \putAnchor{(-4,5)}{E2PriO}
 \draw (1,5) rectangle (5,8) node [midway] {$t^\dag$};
 \putAnchor{(2,5)}{TFreeO}
 \putAnchor{(4,5)}{TPriO}
 \drawContractionNode{k}{(-6,2)}{1.5}{1}{}{}
 \drawCutNode{(-1,0)}{}{}
 \drawContractionNode{}{(-13,2)}{1.2}{1}{}{M2} \thickened
 \drawCutNode{(-10.5,-1)}{}{M2} \thickened
 \putAnchor{(0,-1.5)}{Midpoint}
 \putAnchor{(-19.5,10)}{M2I}
 \putAnchor{(-18,3.6)}{M2Midpoint}
 \putAnchor{(-17,10)}{XsI}
 \putAnchor{(-16,4.4)}{XsMidpoint}
 \draw (-16,-6) rectangle (-3,-3) node [midway] {$(E_1)^{\dag\dag}$};
 \putAnchor{(-15,-3)}{E1FreeE2I}
 \putAnchor{(-8,-3)}{E1FreeTI}
 \putAnchor{(-4,-3)}{E1PriI}
 \putAnchor{(-15,-6)}{E1FreeO}
 \putAnchor{(-13,-6)}{E1M1OL}
 \putAnchor{(-8,-6)}{E1M1OR}
 \putAnchor{(-4,-6)}{E1PriO}
 \drawContractionNode{}{(-13,-9)}{1.2}{1}{}{M1} \thickened
 \drawCutNode{(-10.5,-12)}{}{M1} \thickened
 \putAnchor{(-22,10)}{M1I}
 \putAnchor{(-20,-7)}{M1Midpoint}
 \drawStraightOpen{}{E2FreeMI}{(0,2)} \strikedOut \withLabel{west}{\far}{$M$}
 \drawStraightOpen{}{E2PriI}{(0,2)}
 \drawStraightConnection{}{E2FreeO}{E1FreeE2I} \strikedOut
 \drawStraightConnection{}{E2M2OL}{ContractionM2I} \strikedOut
 \drawStraightConnection{}{E2BoundO}{ContractionI} \strikedOut \withLabel{west}{\near}{$\mathbf{x}$}
 \drawStraightConnection{}{E2PriO}{E1PriI}
 \drawCorneredConnection{}{M2Midpoint}{M2I} \strikedOut \withLabel{north}{\far}{$M_2$}
 \drawCorneredConnection{}{M2Midpoint}{ContractionM2IL} \strikedOut
 \drawCorneredConnection{}{CutM2L}{ContractionM2O} \strikedOut
 \drawCorneredConnection{}{CutM2R}{E2M2OR} \strikedOut
 \drawCorneredConnection{}{XsMidpoint}{XsI} \strikedOut \withLabel{east}{\near}{$\mathbf{x}$}
 \drawCorneredConnection{}{XsMidpoint}{ContractionIL}
 \drawCorneredConnection{}{CutL}{ContractionO} \withLabel{east}{\near}{$x$}
 \drawCorneredConnection{}{CutR}{TPriO}
 \drawCorneredConnection{}{Midpoint}{TFreeO} \strikedOut
 \drawCorneredConnection{}{Midpoint}{E1FreeTI}

 \drawStraightConnection{}{E1M1OL}{ContractionM1I} \strikedOut
 \drawCorneredConnection{}{M1Midpoint}{M1I} \strikedOut \withLabel{east}{\near}{$M_1$}
 \drawCorneredConnection{}{M1Midpoint}{ContractionM1IL} \strikedOut
 \drawCorneredConnection{}{CutM1L}{ContractionM1O} \strikedOut
 \drawCorneredConnection{}{CutM1R}{E1M1OR} \strikedOut

 \drawStraightOpen{}{E1FreeO}{(0,-8)} \strikedOut
 \drawStraightOpen{}{E1PriO}{(0,-8)}
\end{tikzpicture}
\fi}

\newcommand{\illustrateSimulationCOneFirst}{\ifshowIllustrateSimulation
\begin{tikzpicture}
 \draw (-6,6) rectangle (-2,9) node [midway] {$t^\dag$};
 \putAnchor{(-5,6)}{TFreeO}
 \putAnchor{(-4,6)}{TPriO}
 \draw (0,6) rectangle (4,9) node [midway] {$u^\dag$};
 \putAnchor{(2,6)}{UFreeO}
 \putAnchor{(3,6)}{UPriO}
 \drawCutNode{(-1,0)}{}{}
 \drawTensorNode{(5,5)}{}{}
 \drawDerelictionNode{(5,2)}{}{}
 \drawAxNode{(9,7)}{}{}
 \midAnchor{TensorIR}{AxL}{Midpoint}
 \draw (-6,-5) rectangle (12,-2) node [midway] {$E^\dag$};
 \putAnchor{(-5,-2)}{ETFreeI}
 \putAnchor{(2,-2)}{EUFreeI}
 \putAnchor{(11,-2)}{EPriI}
 \putAnchor{(11,-5)}{EPriO}
 \drawStraightConnection{}{TFreeO}{ETFreeI} \strikedOut
 \drawStraightConnection{}{UFreeO}{EUFreeI} \strikedOut
 \drawCorneredConnection{}{CutL}{TPriO}
 \drawCorneredConnection{}{CutR}{DerelictionO}
 \drawStraightConnection{}{TensorO}{DerelictionI}
 \drawCorneredConnection{}{TensorIL}{UPriO}
 \drawCorneredConnection{}{TensorIR}{Midpoint}
 \drawCorneredConnection{}{AxL}{Midpoint}
 \drawCorneredConnection{}{AxR}{EPriI} \withUpToken{right}
 \drawStraightOpen{}{EPriO}{(0,-2)}
\end{tikzpicture}
\fi}
\newcommand{\illustrateSimulationCOneSecond}{\ifshowIllustrateSimulation
\begin{tikzpicture}
 \draw (-6,6) rectangle (-2,9) node [midway] {$t^\dag$};
 \putAnchor{(-5,6)}{TFreeO}
 \putAnchor{(-4,6)}{TPriO}
 \draw (0,6) rectangle (4,9) node [midway] {$u^\dag$};
 \putAnchor{(2,6)}{UFreeO}
 \putAnchor{(3,6)}{UPriO}
 \drawCutNode{(-1,0)}{}{}
 \drawTensorNode{(5,5)}{}{}
 \drawDerelictionNode{(5,2)}{}{}
 \drawAxNode{(9,7)}{}{}
 \midAnchor{TensorIR}{AxL}{Midpoint}
 \draw (-6,-5) rectangle (12,-2) node [midway] {$E^\dag$};
 \putAnchor{(-5,-2)}{ETFreeI}
 \putAnchor{(2,-2)}{EUFreeI}
 \putAnchor{(11,-2)}{EPriI}
 \putAnchor{(11,-5)}{EPriO}
 \drawStraightConnection{}{TFreeO}{ETFreeI} \strikedOut
 \drawStraightConnection{}{UFreeO}{EUFreeI} \strikedOut
 \drawCorneredConnection{}{CutL}{TPriO} \withUpToken{right}
 \drawCorneredConnection{}{CutR}{DerelictionO}
 \drawStraightConnection{}{TensorO}{DerelictionI}
 \drawCorneredConnection{}{TensorIL}{UPriO}
 \drawCorneredConnection{}{TensorIR}{Midpoint}
 \drawCorneredConnection{}{AxL}{Midpoint}
 \drawCorneredConnection{}{AxR}{EPriI}
 \drawStraightOpen{}{EPriO}{(0,-2)}
\end{tikzpicture}
\fi}
\newcommand{\illustrateSimulationCTwoFirst}{\ifshowIllustrateSimulation
\begin{tikzpicture}
 \draw (-11,5) rectangle (-3,8) node [midway] {$(E_2)_{[x]}^\dag$};
 \putAnchor{(-10,8)}{E2XI}
 \putAnchor{(-4,8)}{E2PriI}
 \putAnchor{(-10,5)}{E2FreeO}
 \putAnchor{(-7,5)}{E2BoundO}
 \putAnchor{(-4,5)}{E2PriO}
 \draw (0,5) rectangle (5,8) node [midway] {$t^\dag$};
 \putAnchor{(2,8)}{TFreeI}
 \putAnchor{(4,8)}{TPriI}
 \putAnchor{(2,5)}{TFreeO}
 \putAnchor{(4,5)}{TPriO}
 \drawAxNode{(-7,10)}{}{}
 \drawContractionNode{}{(-7,3)}{1}{1}{}{}
 \drawCutNode{(-1,1)}{}{}
 \putAnchor{(0,-0.5)}{Midpoint}
 \draw (-11,-5) rectangle (-3,-2) node [midway] {$E_1^\dag$};
 \putAnchor{(-10,-2)}{E1E2FreeI}
 \putAnchor{(-7,-2)}{E1TFreeI}
 \putAnchor{(-4,-2)}{E1PriI}
 \putAnchor{(-4,-5)}{E1PriO}
 \drawCorneredConnection{}{AxL}{E2XI}
 \drawCorneredConnection{}{AxR}{E2PriI} \withUpToken{right}
 \drawStraightConnection{}{E2FreeO}{E1E2FreeI} \strikedOut
 \drawStraightConnection{}{ContractionI}{E2BoundO} \strikedOut
 \drawStraightConnection{}{E2PriO}{E1PriI}
 \drawCorneredConnection{}{CutL}{ContractionO}
 \drawCorneredConnection{}{CutR}{TPriO}
 \drawCorneredConnection{}{Midpoint}{TFreeO} \strikedOut
 \drawCorneredConnection{}{Midpoint}{E1TFreeI}
 \drawStraightOpen{}{E1PriO}{(0,-2)}
\end{tikzpicture}
\fi}
\newcommand{\illustrateSimulationCTwoSecond}{\ifshowIllustrateSimulation
\begin{tikzpicture}
 \draw (-11,5) rectangle (-3,8) node [midway] {$(E_2)_\emptyset^\dag$};
 \putAnchor{(-4,8)}{E2PriI}
 \putAnchor{(-10,5)}{E2FreeO}
 \putAnchor{(-7,5)}{E2BoundO}
 \putAnchor{(-4,5)}{E2PriO}
 \draw (0,5) rectangle (5,8) node [midway] {$t^\dag$};
 \putAnchor{(2,8)}{TFreeI}
 \putAnchor{(4,8)}{TPriI}
 \putAnchor{(2,5)}{TFreeO}
 \putAnchor{(4,5)}{TPriO}
 \drawAxNode{(-8,10)}{}{}
 \putAnchor{(-12,5)}{XUpperPoint}
 \midAnchor{XUpperPoint}{ContractionIL}{XLowerPoint}
 \drawContractionNode{}{(-7,2)}{1}{1}{}{}
 \drawCutNode{(-1,0)}{}{}
 \putAnchor{(0,-1.5)}{Midpoint}
 \draw (-11,-6) rectangle (-3,-3) node [midway] {$E_1^\dag$};
 \putAnchor{(-10,-3)}{E1E2FreeI}
 \putAnchor{(-7,-3)}{E1TFreeI}
 \putAnchor{(-4,-3)}{E1PriI}
 \putAnchor{(-4,-6)}{E1PriO}
 \drawCorneredConnection{}{AxL}{XUpperPoint}
 \drawCorneredConnection{}{AxR}{E2PriI} \withUpToken{right}
 \drawStraightConnection{}{E2FreeO}{E1E2FreeI} \strikedOut
 \drawCorneredConnection{}{XLowerPoint}{XUpperPoint}
 \drawCorneredConnection{}{XLowerPoint}{ContractionIL}
 \drawStraightConnection{}{ContractionI}{E2BoundO} \strikedOut
 \drawStraightConnection{}{E2PriO}{E1PriI}
 \drawCorneredConnection{}{CutL}{ContractionO}
 \drawCorneredConnection{}{CutR}{TPriO}
 \drawCorneredConnection{}{Midpoint}{TFreeO} \strikedOut
 \drawCorneredConnection{}{Midpoint}{E1TFreeI}
 \drawStraightOpen{}{E1PriO}{(0,-2)}
\end{tikzpicture}
\fi}
\newcommand{\illustrateSimulationCTwoThird}{\ifshowIllustrateSimulation
\begin{tikzpicture}
 \draw (-11,5) rectangle (-3,8) node [midway] {$(E_2)_\emptyset^\dag$};
 \putAnchor{(-4,8)}{E2PriI}
 \putAnchor{(-10,5)}{E2FreeO}
 \putAnchor{(-7,5)}{E2BoundO}
 \putAnchor{(-4,5)}{E2PriO}
 \draw (0,5) rectangle (5,8) node [midway] {$t^\dag$};
 \putAnchor{(2,8)}{TFreeI}
 \putAnchor{(4,8)}{TPriI}
 \putAnchor{(2,5)}{TFreeO}
 \putAnchor{(4,5)}{TPriO}
 \drawAxNode{(-8,10)}{}{}
 \putAnchor{(-12,5)}{XUpperPoint}
 \midAnchor{XUpperPoint}{ContractionIL}{XLowerPoint}
 \drawContractionNode{}{(-7,2)}{1}{1}{}{}
 \drawCutNode{(-1,0)}{}{}
 \putAnchor{(0,-1.5)}{Midpoint}
 \draw (-11,-6) rectangle (-3,-3) node [midway] {$E_1^\dag$};
 \putAnchor{(-10,-3)}{E1E2FreeI}
 \putAnchor{(-7,-3)}{E1TFreeI}
 \putAnchor{(-4,-3)}{E1PriI}
 \putAnchor{(-4,-6)}{E1PriO}
 \drawCorneredConnection{}{AxL}{XUpperPoint}
 \drawCorneredConnection{}{AxR}{E2PriI}
 \drawStraightConnection{}{E2FreeO}{E1E2FreeI} \strikedOut
 \drawCorneredConnection{}{XLowerPoint}{XUpperPoint}
 \drawCorneredConnection{}{XLowerPoint}{ContractionIL}
 \drawStraightConnection{}{ContractionI}{E2BoundO} \strikedOut
 \drawStraightConnection{}{E2PriO}{E1PriI}
 \drawCorneredConnection{}{CutL}{ContractionO}
 \drawCorneredConnection{}{CutR}{TPriO} \withUpToken{right}
 \drawCorneredConnection{}{Midpoint}{TFreeO} \strikedOut
 \drawCorneredConnection{}{Midpoint}{E1TFreeI}
 \drawStraightOpen{}{E1PriO}{(0,-2)}
\end{tikzpicture}
\fi}
\newcommand{\illustrateSimulationCThreeFirst}{\ifshowIllustrateSimulation
\begin{tikzpicture}
 \draw (-4,9) rectangle (4,12) node [midway] {$v^\diamond$};
 \putAnchor{(-3,9)}{TFreeO}
 \putAnchor{(3,9)}{TPriO}
 \draw [dashed] (-5,6) rectangle (5,13);
 \draw (-4,0) rectangle (4,3) node [midway] {$E_{\FV(v)}^\dag$};
 \putAnchor{(-3,3)}{EFreeI}
 \putAnchor{(3,3)}{EPriI}
 \putAnchor{(3,0)}{EPriO}
 \drawAuxiliaryNode{(-3,6)}{}{} \thickened
 \drawPrincipalNode{(3,6)}{}{}
 \drawStraightConnection{}{TFreeO}{AuxiliaryI} \strikedOut
 \drawStraightConnection{}{TPriO}{PrincipalI}
 \drawStraightConnection{}{AuxiliaryO}{EFreeI} \strikedOut
 \drawStraightConnection{}{PrincipalO}{EPriI} \withUpToken{left}
 \drawStraightOpen{}{EPriO}{(0,-2)}
\end{tikzpicture}
\fi}
\newcommand{\illustrateSimulationCThreeSecond}{\ifshowIllustrateSimulation
\begin{tikzpicture}
 \draw (-4,9) rectangle (4,12) node [midway] {$v^\diamond$};
 \putAnchor{(-3,9)}{TFreeO}
 \putAnchor{(3,9)}{TPriO}
 \draw [dashed] (-5,6) rectangle (5,13);
 \draw (-4,0) rectangle (4,3) node [midway] {$E_{\FV(v)}^\dag$};
 \putAnchor{(-3,3)}{EFreeI}
 \putAnchor{(3,3)}{EPriI}
 \putAnchor{(3,0)}{EPriO}
 \drawAuxiliaryNode{(-3,6)}{}{} \thickened
 \drawPrincipalNode{(3,6)}{}{}
 \drawStraightConnection{}{TFreeO}{AuxiliaryI} \strikedOut
 \drawStraightConnection{}{TPriO}{PrincipalI} \withUpToken{left}
 \drawStraightConnection{}{AuxiliaryO}{EFreeI} \strikedOut
 \drawStraightConnection{}{PrincipalO}{EPriI}
 \drawStraightOpen{}{EPriO}{(0,-2)}
\fi\end{tikzpicture}
}
\newcommand{\illustrateSimulationMFirst}{\ifshowIllustrateSimulation
nnn\begin{tikzpicture}
 \draw (-8,19) rectangle (-1,22) node [midway] {$t^\dag$};
 \putAnchor{(-7,19)}{TFreeO}
 \putAnchor{(-5.5,19)}{TBoundO}
 \putAnchor{(-2,19)}{TPriO}
 \drawContractionNode{}{(-5.5,17)}{1}{1}{}{}
 \drawParrNode{(-3.5,15)}{}{}
 \draw [dashed] (-9,11) rectangle (0,23);
 \drawAuxiliaryNode{(-7,11)}{}{} \thickened
 \drawPrincipalNode{(-3.5,11)}{}{}
 \draw (-8,6) rectangle (-3,9) node [midway] {$A^\dag$};
 \putAnchor{(-7,9)}{AFreeI}
 \putAnchor{(-3.5,9)}{APriI}
 \putAnchor{(-7,6)}{AFreeO}
 \putAnchor{(-3.5,6)}{APriO}
 \draw (0,6) rectangle (4,9) node [midway] {$u^\dag$};
 \putAnchor{(2,6)}{UFreeO}
 \putAnchor{(3,6)}{UPriO}
 \drawCutNode{(-1,0)}{}{}
 \drawTensorNode{(5,5)}{}{}
 \drawDerelictionNode{(5,2)}{}{}
 \drawAxNode{(9,7)}{}{}
 \midAnchor{TensorIR}{AxL}{Midpoint}
 \draw (-8,-5) rectangle (12,-2) node [midway] {$E^\dag$};
 \putAnchor{(-7,-2)}{EAFreeI}
 \putAnchor{(2,-2)}{EUFreeI}
 \putAnchor{(11,-2)}{EPriI}
 \putAnchor{(11,-5)}{EPriO}
 \drawStraightConnection{}{TFreeO}{AuxiliaryI} \strikedOut
 \drawStraightConnection{}{AuxiliaryO}{AFreeI} \strikedOut
 \drawStraightConnection{}{TBoundO}{ContractionI} \strikedOut
 \drawCorneredConnection{}{ParrIL}{ContractionO}
 \drawCorneredConnection{}{ParrIR}{TPriO}
 \drawStraightConnection{}{ParrO}{PrincipalI} \withUpToken{right}
 \drawStraightConnection{}{PrincipalO}{APriI}
 \drawStraightConnection{}{AFreeO}{EAFreeI} \strikedOut
 \drawStraightConnection{}{UFreeO}{EUFreeI} \strikedOut
 \drawCorneredConnection{}{CutL}{APriO}
 \drawCorneredConnection{}{CutR}{DerelictionO}
 \drawStraightConnection{}{TensorO}{DerelictionI}
 \drawCorneredConnection{}{TensorIL}{UPriO}
 \drawCorneredConnection{}{TensorIR}{Midpoint}
 \drawCorneredConnection{}{AxL}{Midpoint}
 \drawCorneredConnection{}{AxR}{EPriI}
 \drawStraightOpen{}{EPriO}{(0,-2)}
\end{tikzpicture}
\fi}
\newcommand{\illustrateSimulationMSecond}{\ifshowIllustrateSimulation
\begin{tikzpicture}
 \draw (-8,19) rectangle (-1,22) node [midway] {$t^\dag$};
 \putAnchor{(-7,19)}{TFreeO}
 \putAnchor{(-5.5,19)}{TBoundO}
 \putAnchor{(-2,19)}{TPriO}
 \drawContractionNode{}{(-5.5,17)}{1}{1}{}{}
 \drawParrNode{(-3.5,15)}{}{}
 \draw [dashed] (-9,11) rectangle (0,23);
 \drawAuxiliaryNode{(-7,11)}{}{} \thickened
 \drawPrincipalNode{(-3.5,11)}{}{}
 \draw (-8,6) rectangle (-5,9) node [midway] {$A^\ddag$};
 \putAnchor{(-7,9)}{AFreeI}
 \putAnchor{(-7,6)}{AFreeO}
 \draw (0,6) rectangle (4,9) node [midway] {$u^\dag$};
 \putAnchor{(2,6)}{UFreeO}
 \putAnchor{(3,6)}{UPriO}
 \drawCutNode{(-1,0)}{}{}
 \drawTensorNode{(5,5)}{}{}
 \drawDerelictionNode{(5,2)}{}{}
 \drawAxNode{(9,7)}{}{}
 \midAnchor{TensorIR}{AxL}{Midpoint}
 \draw (-8,-5) rectangle (12,-2) node [midway] {$E^\dag$};
 \putAnchor{(-7,-2)}{EAFreeI}
 \putAnchor{(2,-2)}{EUFreeI}
 \putAnchor{(11,-2)}{EPriI}
 \putAnchor{(11,-5)}{EPriO}
 \drawStraightConnection{}{TFreeO}{AuxiliaryI} \strikedOut
 \drawStraightConnection{}{AuxiliaryO}{AFreeI} \strikedOut
 \drawStraightConnection{}{TBoundO}{ContractionI} \strikedOut
 \drawCorneredConnection{}{ParrIL}{ContractionO}
 \drawCorneredConnection{}{ParrIR}{TPriO}
 \drawStraightConnection{}{ParrO}{PrincipalI} \withUpToken{right}
 \drawStraightConnection{}{AFreeO}{EAFreeI} \strikedOut
 \drawStraightConnection{}{UFreeO}{EUFreeI} \strikedOut
 \drawCorneredConnection{}{CutL}{PrincipalO}
 \drawCorneredConnection{}{CutR}{DerelictionO}
 \drawStraightConnection{}{TensorO}{DerelictionI}
 \drawCorneredConnection{}{TensorIL}{UPriO}
 \drawCorneredConnection{}{TensorIR}{Midpoint}
 \drawCorneredConnection{}{AxL}{Midpoint}
 \drawCorneredConnection{}{AxR}{EPriI}
 \drawStraightOpen{}{EPriO}{(0,-2)}
\end{tikzpicture}
\fi}
\newcommand{\illustrateSimulationMThird}{\ifshowIllustrateSimulation
\begin{tikzpicture}
 \draw (-8,13) rectangle (-1,16) node [midway] {$t^\dag$};
 \putAnchor{(-7,13)}{TFreeO}
 \putAnchor{(-5,13)}{TBoundO}
 \putAnchor{(-3,13)}{TPriO}
 \drawContractionNode{}{(-5,11)}{1}{1}{}{}
 \draw (-8,6) rectangle (-5,9) node [midway] {$A^\ddag$};
 \putAnchor{(-7,9)}{AFreeI}
 \putAnchor{(-7,6)}{AFreeO}
 \draw (0,13) rectangle (4,16) node [midway] {$u^\dag$};
 \putAnchor{(2,13)}{UFreeO}
 \putAnchor{(3,13)}{UPriO}
 \drawCutNode{(-0.5,9)}{}{}
 \draw (-8,2) rectangle (-1,5) node [midway] {$E^\dag$};
 \putAnchor{(-7,5)}{EAFreeI}
 \putAnchor{(-4,5)}{EUFreeI}
 \putAnchor{(-3,5)}{EPriI}
 \putAnchor{(-3,2)}{EPriO}
 \putAnchor{(-0.5,7)}{Midpoint}
 \midAnchor{TPriO}{EPriI}{PriMidpoint}
 \drawStraightConnection{}{TFreeO}{AFreeI} \strikedOut
 \drawStraightConnection{}{TBoundO}{ContractionI} \strikedOut
 \drawStraightConnection{}{TPriO}{PriMidpoint} \withUpToken{right}
 \drawStraightConnection{}{PriMidpoint}{EPriI}
 \drawStraightConnection{}{AFreeO}{EAFreeI} \strikedOut
 \drawCorneredConnection{}{CutL}{ContractionO}
 \drawCorneredConnection{}{CutR}{UPriO}
 \drawCorneredConnection{}{Midpoint}{UFreeO} \strikedOut
 \drawCorneredConnection{}{Midpoint}{EUFreeI}
 \drawStraightOpen{}{EPriO}{(0,-2)}
\end{tikzpicture}
\fi}
\newcommand{\illustrateSimulationMFourth}{\ifshowIllustrateSimulation
\begin{tikzpicture}
 \draw (-8,13) rectangle (-1,16) node [midway] {$t^\dag$};
 \putAnchor{(-7,13)}{TFreeO}
 \putAnchor{(-5,13)}{TBoundO}
 \putAnchor{(-3,13)}{TPriO}
 \drawContractionNode{}{(-5,11)}{1}{1}{}{}
 \draw (-8,3) rectangle (-1,6) node [midway] {$A^\dag$};
 \putAnchor{(-7,6)}{ATFreeI}
 \putAnchor{(-6,6)}{AUFreeI}
 \putAnchor{(-3,6)}{APriI}
 \putAnchor{(-7,3)}{AFreeO}
 \putAnchor{(-3,3)}{APriO}
 \draw (0,13) rectangle (4,16) node [midway] {$u^\dag$};
 \putAnchor{(2,13)}{UFreeO}
 \putAnchor{(3,13)}{UPriO}
 \drawCutNode{(-0.5,9)}{}{}
 \draw (-8,-1) rectangle (-1,2) node [midway] {$E^\dag$};
 \putAnchor{(-7,2)}{EFreeI}
 \putAnchor{(-3,2)}{EPriI}
 \putAnchor{(-3,-1)}{EPriO}
 \putAnchor{(-0.5,7)}{Midpoint}
 \midAnchor{TPriO}{APriI}{PriMidpoint}
 \drawStraightConnection{}{TFreeO}{ATFreeI} \strikedOut
 \drawStraightConnection{}{TBoundO}{ContractionI} \strikedOut
 \drawStraightConnection{}{TPriO}{PriMidpoint} \withUpToken{right}
 \drawStraightConnection{}{PriMidpoint}{APriI}
 \drawCorneredConnection{}{Midpoint}{UFreeO} \strikedOut
 \drawCorneredConnection{}{Midpoint}{AUFreeI}
 \drawCorneredConnection{}{CutL}{ContractionO}
 \drawCorneredConnection{}{CutR}{UPriO}
 \drawStraightConnection{}{AFreeO}{EFreeI} \strikedOut
 \drawStraightConnection{}{APriO}{EPriI}
 \drawStraightOpen{}{EPriO}{(0,-2)}
\end{tikzpicture}
\fi}
\newcommand{\illustrateSimulationEOneFirst}{\ifshowIllustrateSimulation
\begin{tikzpicture}
 \draw (-11,5) rectangle (-3,8) node [midway] {$(E_2)_\emptyset^\dag$};
 \putAnchor{(-4,8)}{E2PriI}
 \putAnchor{(-10,5)}{E2FreeO}
 \putAnchor{(-7,5)}{E2BoundO}
 \putAnchor{(-4,5)}{E2PriO}
 \draw (1,9) rectangle (6,12) node [midway] {$v^\diamond$};
 \putAnchor{(2,9)}{VFreeO}
 \putAnchor{(5,9)}{VPriO}
 \draw [dashed] (0,6) rectangle (7,13);
 \drawAuxiliaryNode{(2,6)}{}{} \thickened
 \drawPrincipalNode{(5,6)}{}{}
 \draw (1,1) rectangle (6,4) node [midway] {$A^\dag$};
 \putAnchor{(2,4)}{AFreeI}
 \putAnchor{(5,4)}{APriI}
 \putAnchor{(2,1)}{AFreeO}
 \putAnchor{(5,1)}{APriO}
 \drawAxNode{(-8,10)}{}{}
 \putAnchor{(-12,5)}{XUpperPoint}
 \drawContractionNode{k+1}{(-7,2)}{2.5}{1}{}{}
 \midAnchor{XUpperPoint}{ContractionIL}{XLowerPoint}
 \drawCutNode{(-1,0)}{}{}
 \putAnchor{(0,-1.5)}{Midpoint}
 \draw (-11,-6) rectangle (-3,-3) node [midway] {$E_1^\dag$};
 \putAnchor{(-10,-3)}{E1E2FreeI}
 \putAnchor{(-7,-3)}{E1TFreeI}
 \putAnchor{(-4,-3)}{E1PriI}
 \putAnchor{(-4,-6)}{E1PriO}
 \drawCorneredConnection{}{AxL}{XUpperPoint}
 \drawCorneredConnection{}{AxR}{E2PriI}
 \drawStraightConnection{}{E2FreeO}{E1E2FreeI} \strikedOut
 \drawCorneredConnection{}{XLowerPoint}{XUpperPoint}
 \drawCorneredConnection{}{XLowerPoint}{ContractionIL}
 \drawStraightConnection{}{ContractionI}{E2BoundO} \strikedOut
 \drawStraightConnection{}{E2PriO}{E1PriI}
 \drawCorneredConnection{}{CutL}{ContractionO}
 \drawStraightConnection{}{VFreeO}{AuxiliaryI} \strikedOut
 \drawStraightConnection{}{AuxiliaryO}{AFreeI} \strikedOut
 \drawStraightConnection{}{VPriO}{PrincipalI} \withUpToken{right}
 \drawStraightConnection{}{PrincipalO}{APriI}
 \drawCorneredConnection{}{CutR}{APriO}
 \drawCorneredConnection{}{Midpoint}{AFreeO} \strikedOut
 \drawCorneredConnection{}{Midpoint}{E1TFreeI}
 \drawStraightOpen{}{E1PriO}{(0,-2)}
\end{tikzpicture}
\fi}
\newcommand{\illustrateSimulationEOneSecond}{\ifshowIllustrateSimulation
\begin{tikzpicture}
 \draw (-11,5) rectangle (-3,8) node [midway] {$(E_2)_\emptyset^\dag$};
 \putAnchor{(-4,8)}{E2PriI}
 \putAnchor{(-10,5)}{E2FreeO}
 \putAnchor{(-7,5)}{E2BoundO}
 \putAnchor{(-4,5)}{E2PriO}
 \draw (1,9) rectangle (6,12) node [midway] {$v^\diamond$};
 \putAnchor{(2,9)}{VFreeO}
 \putAnchor{(5,9)}{VPriO}
 \draw [dashed] (0,6) rectangle (7,13);
 \drawAuxiliaryNode{(2,6)}{}{} \thickened
 \drawPrincipalNode{(5,6)}{}{}
 \draw (1,1) rectangle (4,4) node [midway] {$A^\ddag$};
 \putAnchor{(2,4)}{AFreeI}
 \putAnchor{(2,1)}{AFreeO}
 \drawAxNode{(-8,10)}{}{}
 \putAnchor{(-12,5)}{XUpperPoint}
 \midAnchor{XUpperPoint}{ContractionIL}{XLowerPoint}
 \drawContractionNode{k+1}{(-7,2)}{2.5}{1}{}{}
 \drawCutNode{(-1,0)}{}{}
 \putAnchor{(0,-1.5)}{Midpoint}
 \draw (-11,-6) rectangle (-3,-3) node [midway] {$E_1^\dag$};
 \putAnchor{(-10,-3)}{E1E2FreeI}
 \putAnchor{(-7,-3)}{E1TFreeI}
 \putAnchor{(-4,-3)}{E1PriI}
 \putAnchor{(-4,-6)}{E1PriO}
 \drawCorneredConnection{}{AxL}{XUpperPoint}
 \drawCorneredConnection{}{AxR}{E2PriI}
 \drawStraightConnection{}{E2FreeO}{E1E2FreeI} \strikedOut
 \drawCorneredConnection{}{XLowerPoint}{XUpperPoint}
 \drawCorneredConnection{}{XLowerPoint}{ContractionIL}
 \drawStraightConnection{}{ContractionI}{E2BoundO} \strikedOut
 \drawStraightConnection{}{E2PriO}{E1PriI}
 \drawCorneredConnection{}{CutL}{ContractionO}
 \drawStraightConnection{}{VFreeO}{AuxiliaryI} \strikedOut
 \drawStraightConnection{}{AuxiliaryO}{AFreeI} \strikedOut
 \drawStraightConnection{}{VPriO}{PrincipalI} \withUpToken{right}
 \drawCorneredConnection{}{CutR}{PrincipalO}
 \drawCorneredConnection{}{Midpoint}{AFreeO} \strikedOut
 \drawCorneredConnection{}{Midpoint}{E1TFreeI}
 \drawStraightOpen{}{E1PriO}{(0,-2)}
\end{tikzpicture}
\fi}
\newcommand{\illustrateSimulationEOneThird}{\ifshowIllustrateSimulation
\begin{tikzpicture}
 \draw (-11,5) rectangle (-3,8) node [midway] {$(E_2)_\emptyset^\dag$};
 \putAnchor{(-4,8)}{E2PriI}
 \putAnchor{(-10,5)}{E2FreeO}
 \putAnchor{(-7,5)}{E2BoundO}
 \putAnchor{(-4,5)}{E2PriO}
 \draw (1,5) rectangle (6,8) node [midway] {$v^\diamond$};
 \putAnchor{(2,5)}{VFreeO}
 \putAnchor{(5,5)}{VPriO}
 \draw [dashed] (0,2) rectangle (7,9);
 \drawAuxiliaryNode{(2,2)}{}{} \thickened
 \drawPrincipalNode{(5,2)}{}{}
 \draw (-11,-6) rectangle (-5,-3) node [midway] {$A^\ddag$};
 \putAnchor{(-10,-3)}{AE1FreeI}
 \putAnchor{(-7,-3)}{AVFreeI}
 \putAnchor{(-10,-6)}{AFreeO}
 \drawAxNode{(-8,10)}{}{}
 \putAnchor{(-12,5)}{XUpperPoint}
 \midAnchor{XUpperPoint}{ContractionIL}{XLowerPoint}
 \drawContractionNode{k+1}{(-7,2)}{2.5}{1}{}{}
 \drawCutNode{(-1,0)}{}{}
 \putAnchor{(0,-1.5)}{Midpoint}
 \draw (-11,-10) rectangle (-3,-7) node [midway] {$E_1^\dag$};
 \putAnchor{(-10,-7)}{E1FreeI}
 \putAnchor{(-4,-7)}{E1PriI}
 \putAnchor{(-4,-10)}{E1PriO}
 \drawCorneredConnection{}{AxL}{XUpperPoint}
 \drawCorneredConnection{}{AxR}{E2PriI}
 \drawStraightConnection{}{E2FreeO}{E1E2FreeI} \strikedOut
 \drawCorneredConnection{}{XLowerPoint}{XUpperPoint}
 \drawCorneredConnection{}{XLowerPoint}{ContractionIL}
 \drawStraightConnection{}{ContractionI}{E2BoundO} \strikedOut
 \drawStraightConnection{}{E2PriO}{E1PriI}
 \drawCorneredConnection{}{CutL}{ContractionO}
 \drawCorneredConnection{}{CutR}{PrincipalO}
 \drawStraightConnection{}{VFreeO}{AuxiliaryI} \strikedOut
 \drawStraightConnection{}{VPriO}{PrincipalI} \withUpToken{right}
 \drawCorneredConnection{}{Midpoint}{AuxiliaryO} \strikedOut
 \drawCorneredConnection{}{Midpoint}{AVFreeI}
 \drawStraightConnection{}{AFreeO}{E1FreeI} \strikedOut
 \drawStraightOpen{}{E1PriO}{(0,-2)}
\end{tikzpicture}
\fi}
\newcommand{\illustrateSimulationEOneFourth}{\ifshowIllustrateSimulation
\begin{tikzpicture}
 \draw (-11,5) rectangle (-3,8) node [midway] {$(E_2)_\emptyset^\dag$};
 \putAnchor{(-4,8)}{E2PriI}
 \putAnchor{(-10,5)}{E2FreeO}
 \putAnchor{(-7,5)}{E2BoundO}
 \putAnchor{(-4,5)}{E2PriO}
 \draw (1,5) rectangle (6,8) node [midway] {$v^\diamond$};
 \putAnchor{(2,5)}{VFreeO}
 \putAnchor{(5,5)}{VPriO}
 \draw [dashed] (0,2) rectangle (7,9);
 \drawAuxiliaryNode{(2,2)}{}{} \thickened
 \drawPrincipalNode{(5,2)}{}{}
 \draw (-11,-4) rectangle (-3,-1) node [midway] {$\plug{E}{A}^{\dag\dag}$};
 \putAnchor{(-10,-1)}{PlugFreeI}
 \putAnchor{(-4,-1)}{PlugPriI}
 \putAnchor{(-9,-4)}{PlugCapturedOL}
 \putAnchor{(-5,-4)}{PlugCapturedOR}
 \putAnchor{(-4,-4)}{PlugPriO}
 \drawAxNode{(-8,10)}{}{}
 \putAnchor{(-12,5)}{XUpperPoint}
 \midAnchor{XUpperPoint}{ContractionIL}{XLowerPoint}
 \drawContractionNode{k+1}{(-7,2)}{2.5}{1}{}{}
 \drawCutNode{(-1,0)}{}{}
 \putAnchor{(0,-5)}{Midpoint}
 \drawContractionNode{}{(-9,-7)}{1.2}{1}{}{Captured} \thickened
 \drawCutNode{(-7,-10)}{}{Captured} \thickened
 \drawCorneredConnection{}{AxL}{XUpperPoint}
 \drawCorneredConnection{}{AxR}{E2PriI}
 \drawStraightConnection{}{E2FreeO}{PlugFreeI} \strikedOut
 \drawCorneredConnection{}{XLowerPoint}{XUpperPoint}
 \drawCorneredConnection{}{XLowerPoint}{ContractionIL}
 \drawStraightConnection{}{ContractionI}{E2BoundO} \strikedOut
 \drawStraightConnection{}{E2PriO}{PlugPriI}
 \drawCorneredConnection{}{CutL}{ContractionO}
 \drawCorneredConnection{}{CutR}{PrincipalO}
 \drawStraightConnection{}{VFreeO}{AuxiliaryI} \strikedOut
 \drawStraightConnection{}{VPriO}{PrincipalI} \withUpToken{right}
 \drawCorneredConnection{}{Midpoint}{AuxiliaryO} \strikedOut
 \drawCorneredConnection{}{Midpoint}{ContractionCapturedIR}
 \drawStraightConnection{}{PlugCapturedOL}{ContractionCapturedI} \strikedOut
 \drawCorneredConnection{}{CutCapturedL}{ContractionCapturedO} \strikedOut
 \drawCorneredConnection{}{CutCapturedR}{PlugCapturedOR} \strikedOut
 \drawStraightOpen{}{PlugPriO}{(0,-8)}
\end{tikzpicture}
\fi}
\newcommand{\illustrateSimulationEOneFifth}{\ifshowIllustrateSimulation
\begin{tikzpicture}
 \draw (-11,4) rectangle (-3,7) node [midway] {$(E_2)_\emptyset^\dag$};
 \putAnchor{(-4,7)}{E2PriI}
 \putAnchor{(-10,4)}{E2FreeO}
 \putAnchor{(-7,4)}{E2BoundO}
 \putAnchor{(-4,4)}{E2PriO}
 \draw (-8.5,13) rectangle (-2.5,16) node [midway] {$(v^\diamond)^\approx$};
 \putAnchor{(-7,13)}{VCopyFreeO}
 \putAnchor{(-4,13)}{VCopyPriO}
 \draw [dashed] (-9,10) rectangle (-2,17);
 \drawAuxiliaryNode{(-7,10)}{}{Copy} \thickened
 \drawPrincipalNode{(-4,10)}{}{Copy}
 \draw (1,4) rectangle (6,7) node [midway] {$v^\diamond$};
 \putAnchor{(2,4)}{VFreeO}
 \putAnchor{(5,4)}{VPriO}
 \draw [dashed] (0,2) rectangle (7,8);
 \drawAuxiliaryNode{(2,2)}{}{} \thickened
 \drawPrincipalNode{(5,2)}{}{}
 \draw (-11,-4) rectangle (-3,-1) node [midway] {$\plug{E}{A}^{\dag\dag}$};
 \putAnchor{(-10,-1)}{PlugFreeI}
 \putAnchor{(-4,-1)}{PlugPriI}
 \putAnchor{(-9,-4)}{PlugCapturedOL}
 \putAnchor{(-5,-4)}{PlugCapturedOR}
 \putAnchor{(-4,-4)}{PlugPriO}
 \drawContractionNode{k}{(-7,2)}{1.5}{1}{}{}
 \drawCutNode{(-1,0)}{}{}
 \putAnchor{(0,-5)}{Midpoint}
 \putAnchor{(-11,8)}{CopyUpperPoint}
 \putAnchor{(-11,-5)}{CopyLowerPoint}
 \putAnchor{(-12,2)}{CopyMidpoint}
 \drawContractionNode{}{(-9,-7)}{1.2}{1}{}{Captured} \thickened
 \drawCutNode{(-7,-10)}{}{Captured} \thickened
 \drawStraightConnection{}{VCopyFreeO}{AuxiliaryCopyI} \strikedOut
 \drawStraightConnection{}{VCopyPriO}{PrincipalCopyI} \withUpToken{right}
 \drawStraightConnection{}{PrincipalCopyO}{E2PriI}
 \drawCorneredConnection{}{CopyUpperPoint}{AuxiliaryCopyO} \strikedOut
 \drawCorneredConnection{}{CopyUpperPoint}{CopyMidpoint}
 \drawStraightConnection{}{E2FreeO}{PlugFreeI} \strikedOut
 \drawStraightConnection{}{ContractionI}{E2BoundO} \strikedOut
 \drawStraightConnection{}{E2PriO}{PlugPriI}
 \drawCorneredConnection{}{CutL}{ContractionO}
 \drawCorneredConnection{}{CutR}{PrincipalO}
 \drawStraightConnection{}{VFreeO}{AuxiliaryI} \strikedOut
 \drawStraightConnection{}{VPriO}{PrincipalI}
 \drawCorneredConnection{}{Midpoint}{AuxiliaryO} \strikedOut
 \drawCorneredConnection{}{Midpoint}{ContractionCapturedIR}
 \drawCorneredConnection{}{CopyLowerPoint}{CopyMidpoint} \strikedOut
 \drawCorneredConnection{}{CopyLowerPoint}{ContractionCapturedIL}
 \drawStraightConnection{}{PlugCapturedOL}{ContractionCapturedI} \strikedOut
 \drawCorneredConnection{}{CutCapturedL}{ContractionCapturedO} \strikedOut
 \drawCorneredConnection{}{CutCapturedR}{PlugCapturedOR} \strikedOut
 \drawStraightOpen{}{PlugPriO}{(0,-8)}
\end{tikzpicture}
\fi}
\newcommand{\illustrateSimulationEOneSixth}{\ifshowIllustrateSimulation
\begin{tikzpicture}
 \draw (-11,4) rectangle (-3,7) node [midway] {$(E_2)_\emptyset^\dag$};
 \putAnchor{(-7,7)}{E2FreeI}
 \putAnchor{(-4,7)}{E2PriI}
 \putAnchor{(-10,4)}{E2FreeO}
 \putAnchor{(-7,4)}{E2BoundO}
 \putAnchor{(-4,4)}{E2PriO}
 \draw (-8.5,12) rectangle (-2.5,15) node [midway] {$(v^\diamond)^\approx$};
 \putAnchor{(-7,12)}{VCopyFreeO}
 \putAnchor{(-4,12)}{VCopyPriO}
 \draw [dashed] (-9,9) rectangle (-2,16);
 \drawAuxiliaryNode{(-7,9)}{}{Copy} \thickened
 \drawPrincipalNode{(-4,9)}{}{Copy}
 \draw (1,4) rectangle (6,7) node [midway] {$v^\diamond$};
 \putAnchor{(2,4)}{VFreeO}
 \putAnchor{(5,4)}{VPriO}
 \draw [dashed] (0,2) rectangle (7,8);
 \drawAuxiliaryNode{(2,2)}{}{} \thickened
 \drawPrincipalNode{(5,2)}{}{}
 \draw (-11,-6) rectangle (-3,-3) node [midway] {$\plug{E}{A}^\dag$};
 \putAnchor{(-10,-3)}{PlugE2FreeI}
 \putAnchor{(-7,-3)}{PlugVFreeI}
 \putAnchor{(-4,-3)}{PlugPriI}
 \putAnchor{(-4,-6)}{PlugPriO}
 \drawContractionNode{k}{(-7,2)}{1.5}{1}{}{}
 \drawCutNode{(-1,0)}{}{}
 \putAnchor{(-1,-2)}{Midpoint}
 \drawStraightConnection{}{VCopyFreeO}{AuxiliaryCopyI} \strikedOut
 \drawStraightConnection{}{VCopyPriO}{PrincipalCopyI} \withUpToken{right}
 \drawStraightConnection{}{PrincipalCopyO}{E2PriI}
 \drawStraightConnection{}{AuxiliaryCopyO}{E2FreeI}
 \drawStraightConnection{}{E2FreeO}{PlugE2FreeI} \strikedOut
 \drawStraightConnection{}{ContractionI}{E2BoundO} \strikedOut
 \drawStraightConnection{}{E2PriO}{PlugPriI}
 \drawCorneredConnection{}{CutL}{ContractionO}
 \drawCorneredConnection{}{CutR}{PrincipalO}
 \drawStraightConnection{}{VFreeO}{AuxiliaryI} \strikedOut
 \drawStraightConnection{}{VPriO}{PrincipalI}
 \drawCorneredConnection{}{Midpoint}{AuxiliaryO} \strikedOut
 \drawCorneredConnection{}{Midpoint}{PlugVFreeI}
 \drawStraightOpen{}{PlugPriO}{(0,-2)}
\end{tikzpicture}
\fi}
\newcommand{\illustrateSimulationETwoFirst}{\ifshowIllustrateSimulation
\begin{tikzpicture}
 \draw (-11,5) rectangle (-3,8) node [midway] {$(E_2)_\emptyset^\dag$};
 \putAnchor{(-4,8)}{E2PriI}
 \putAnchor{(-10,5)}{E2FreeO}
 \putAnchor{(-4,5)}{E2PriO}
 \draw (1,9) rectangle (6,12) node [midway] {$v^\diamond$};
 \putAnchor{(2,9)}{VFreeO}
 \putAnchor{(5,9)}{VPriO}
 \draw [dashed] (0,6) rectangle (7,13);
 \drawAuxiliaryNode{(2,6)}{}{} \thickened
 \drawPrincipalNode{(5,6)}{}{}
 \draw (1,1) rectangle (6,4) node [midway] {$A^\dag$};
 \putAnchor{(2,4)}{AFreeI}
 \putAnchor{(5,4)}{APriI}
 \putAnchor{(2,1)}{AFreeO}
 \putAnchor{(5,1)}{APriO}
 \drawAxNode{(-8,10)}{}{}
 \putAnchor{(-12,5)}{XUpperPoint}
 \midAnchor{XUpperPoint}{ContractionIL}{XLowerPoint}
 \drawContractionNode{1}{(-7,2)}{1.5}{1}{}{}
 \drawCutNode{(-1,0)}{}{}
 \putAnchor{(0,-1.5)}{Midpoint}
 \draw (-11,-6) rectangle (-3,-3) node [midway] {$E_1^\dag$};
 \putAnchor{(-10,-3)}{E1E2FreeI}
 \putAnchor{(-7,-3)}{E1TFreeI}
 \putAnchor{(-4,-3)}{E1PriI}
 \putAnchor{(-4,-6)}{E1PriO}
 \drawCorneredConnection{}{AxL}{XUpperPoint}
 \drawCorneredConnection{}{AxR}{E2PriI}
 \drawStraightConnection{}{E2FreeO}{E1E2FreeI} \strikedOut
 \drawCorneredConnection{}{XLowerPoint}{XUpperPoint}
 \drawCorneredConnection{}{XLowerPoint}{ContractionI}
 \drawStraightConnection{}{E2PriO}{E1PriI}
 \drawCorneredConnection{}{CutL}{ContractionO}
 \drawStraightConnection{}{VFreeO}{AuxiliaryI} \strikedOut
 \drawStraightConnection{}{AuxiliaryO}{AFreeI} \strikedOut
 \drawStraightConnection{}{VPriO}{PrincipalI} \withUpToken{right}
 \drawStraightConnection{}{PrincipalO}{APriI}
 \drawCorneredConnection{}{CutR}{APriO}
 \drawCorneredConnection{}{Midpoint}{AFreeO} \strikedOut
 \drawCorneredConnection{}{Midpoint}{E1TFreeI}
 \drawStraightOpen{}{E1PriO}{(0,-2)}
\end{tikzpicture}
\fi}
\newcommand{\illustrateSimulationETwoSecond}{\ifshowIllustrateSimulation
\begin{tikzpicture}
 \draw (-11,5) rectangle (-3,8) node [midway] {$(E_2)_\emptyset^\dag$};
 \putAnchor{(-4,8)}{E2PriI}
 \putAnchor{(-10,5)}{E2FreeO}
 \putAnchor{(-4,5)}{E2PriO}
 \draw (1,9) rectangle (6,12) node [midway] {$v^\diamond$};
 \putAnchor{(2,9)}{VFreeO}
 \putAnchor{(5,9)}{VPriO}
 \draw [dashed] (0,6) rectangle (7,13);
 \drawAuxiliaryNode{(2,6)}{}{} \thickened
 \drawPrincipalNode{(5,6)}{}{}
 \draw (1,1) rectangle (4,4) node [midway] {$A^\ddag$};
 \putAnchor{(2,4)}{AFreeI}
 \putAnchor{(2,1)}{AFreeO}
 \drawAxNode{(-8,10)}{}{}
 \putAnchor{(-12,5)}{XUpperPoint}
 \midAnchor{XUpperPoint}{ContractionIL}{XLowerPoint}
 \drawContractionNode{1}{(-7,2)}{1.5}{1}{}{}
 \drawCutNode{(-1,0)}{}{}
 \putAnchor{(0,-1.5)}{Midpoint}
 \draw (-11,-6) rectangle (-3,-3) node [midway] {$E_1^\dag$};
 \putAnchor{(-10,-3)}{E1E2FreeI}
 \putAnchor{(-7,-3)}{E1TFreeI}
 \putAnchor{(-4,-3)}{E1PriI}
 \putAnchor{(-4,-6)}{E1PriO}
 \drawCorneredConnection{}{AxL}{XUpperPoint}
 \drawCorneredConnection{}{AxR}{E2PriI}
 \drawStraightConnection{}{E2FreeO}{E1E2FreeI} \strikedOut
 \drawCorneredConnection{}{XLowerPoint}{XUpperPoint}
 \drawCorneredConnection{}{XLowerPoint}{ContractionI}
 \drawStraightConnection{}{E2PriO}{E1PriI}
 \drawCorneredConnection{}{CutL}{ContractionO}
 \drawStraightConnection{}{VFreeO}{AuxiliaryI} \strikedOut
 \drawStraightConnection{}{AuxiliaryO}{AFreeI} \strikedOut
 \drawStraightConnection{}{VPriO}{PrincipalI} \withUpToken{right}
 \drawCorneredConnection{}{CutR}{PrincipalO}
 \drawCorneredConnection{}{Midpoint}{AFreeO} \strikedOut
 \drawCorneredConnection{}{Midpoint}{E1TFreeI}
 \drawStraightOpen{}{E1PriO}{(0,-2)}
\end{tikzpicture}
\fi}
\newcommand{\illustrateSimulationETwoThird}{\ifshowIllustrateSimulation
\begin{tikzpicture}
 \draw (-8,0) rectangle (0,3) node [midway] {$(E_2)_\emptyset^\dag$};
 \putAnchor{(-1,3)}{E2PriI}
 \putAnchor{(-7,0)}{E2FreeO}
 \putAnchor{(-1,0)}{E2PriO}
 \draw (1,9) rectangle (6,12) node [midway] {$v^\diamond$};
 \putAnchor{(2,9)}{VFreeO}
 \putAnchor{(5,9)}{VPriO}
 \draw [dashed] (0,6) rectangle (7,13);
 \drawAuxiliaryNode{(2,6)}{}{} \thickened
 \drawPrincipalNode{(5,6)}{}{}
 \draw (1,0) rectangle (4,3) node [midway] {$A^\ddag$};
 \putAnchor{(2,3)}{AFreeI}
 \putAnchor{(2,0)}{AFreeO}
 \putAnchor{(0,4.5)}{PriMidpoint}
 \putAnchor{(0,-1)}{Midpoint}
 \draw (-8,-5) rectangle (0,-2) node [midway] {$E_1^\dag$};
 \putAnchor{(-7,-2)}{E1E2FreeI}
 \putAnchor{(-4,-2)}{E1TFreeI}
 \putAnchor{(-1,-2)}{E1PriI}
 \putAnchor{(-1,-5)}{E1PriO}
 \drawStraightConnection{}{E2FreeO}{E1E2FreeI} \strikedOut
 \drawStraightConnection{}{E2PriO}{E1PriI}
 \drawStraightConnection{}{VFreeO}{AuxiliaryI} \strikedOut
 \drawStraightConnection{}{AuxiliaryO}{AFreeI} \strikedOut
 \drawStraightConnection{}{VPriO}{PrincipalI} \withUpToken{right}
 \drawCorneredConnection{}{PriMidpoint}{PrincipalO}
 \drawCorneredConnection{}{PriMidpoint}{E2PriI}
 \drawCorneredConnection{}{Midpoint}{AFreeO} \strikedOut
 \drawCorneredConnection{}{Midpoint}{E1TFreeI}
 \drawStraightOpen{}{E1PriO}{(0,-2)}
\end{tikzpicture}
\fi}
\newcommand{\illustrateSimulationETwoFourth}{\ifshowIllustrateSimulation
\begin{tikzpicture}
 \draw (-6,0) rectangle (1,3) node [midway] {$(E_2)_\emptyset^\dag$};
 \putAnchor{(-4,3)}{E2FreeI}
 \putAnchor{(-1,3)}{E2PriI}
 \putAnchor{(-5,0)}{E2FreeOL}
 \putAnchor{(-4,0)}{E2FreeOR}
 \putAnchor{(-1,0)}{E2PriO}
 \draw (-5,8) rectangle (0,11) node [midway] {$v^\diamond$};
 \putAnchor{(-4,8)}{VFreeO}
 \putAnchor{(-1,8)}{VPriO}
 \draw [dashed] (-6,5) rectangle (1,12);
 \drawAuxiliaryNode{(-4,5)}{}{} \thickened
 \drawPrincipalNode{(-1,5)}{}{}
 \draw (-6,-4) rectangle (-2,-1) node [midway] {$A^\ddag$};
 \putAnchor{(-5,-1)}{AFreeIL}
 \putAnchor{(-4,-1)}{AFreeIR}
 \putAnchor{(-5,-4)}{AFreeOL}
 \putAnchor{(-4,-4)}{AFreeOR}
 \draw (-6,-8) rectangle (1,-5) node [midway] {$E_1^\dag$};
 \putAnchor{(-5,-5)}{E1FreeIL}
 \putAnchor{(-4,-5)}{E1FreeIR}
 \putAnchor{(-1,-5)}{E1PriI}
 \putAnchor{(-1,-8)}{E1PriO}
 \drawStraightConnection{}{E2FreeOL}{AFreeIL} \strikedOut
 \drawStraightConnection{}{E2FreeOR}{AFreeIR} \strikedOut
 \drawStraightConnection{}{E2PriO}{E1PriI}
 \drawStraightConnection{}{VFreeO}{AuxiliaryI} \strikedOut
 \drawStraightConnection{}{AuxiliaryO}{E2FreeI} \strikedOut
 \drawStraightConnection{}{VPriO}{PrincipalI} \withUpToken{right}
 \drawStraightConnection{}{PrincipalO}{E2PriI}
 \drawStraightConnection{}{AFreeOL}{E1FreeIL} \strikedOut
 \drawStraightConnection{}{AFreeOR}{E1FreeIR} \strikedOut
 \drawStraightOpen{}{E1PriO}{(0,-2)}
\fi\end{tikzpicture}
}